\documentclass[graybox]{svmult}

\usepackage{makeidx}         
\usepackage{graphicx}        

\makeindex             

\usepackage{amsmath}
\usepackage{amsfonts}
\usepackage{stmaryrd}
\usepackage[all]{xy}
\usepackage{Guyboxes}
\usepackage{mathpartir}
\usepackage{todonotes}
\usepackage{url}

\def\Fin{\mathrm{Var}}
\def\N{\mathbb{N}}

\def\None{{\N_1}}

\def\Ncat{{\cal N}}
\def\Ncatop{{\cal N}^{\mathbf{op}}}
\def\Id{\I}
\def\I{\mathrm{I}}
\def\Iext{\I_\mathrm{ext}}
\def\T{{\cal T}}
\def\B{{\cal B}}
\def\Tlccc{\T^{\Iext,\None,\Sigma,\Pi,o}}

\def\C{{\cal C}}
\def\CF{{\bf C}}
\def\LF{{\bf L}}
\def\Cobj{{\cal C}_0}
\def\D{{\cal D}}
\def\Y{{\mathrm{Y}}}
\def\id{\mathrm{id}}
\def\op{\mathrm{op}}
\def\Cop{\C^\op}
\def\Set{\mathrm{Set}}
\def\Fam{\mathrm{Fam}}
\def\Ty{\mathrm{Ty}}
\def\Tycat{{\cal T}\!\mathit{y}}
\def\Tm{\mathrm{Tm}}

\def\cext{\cdot}
\def\s{\mathrm{s}}
\def\p{\mathrm{p}}
\def\q{\mathrm{q}}
\def\qI{\mathrm{q}}

\def\arrow{\Rightarrow}
\def\ap{\mathrm{ap}}
\def\var{\mathrm{var}}
\def\ev{\mathrm{\varepsilon}}
\def\fst{\mathrm{fst}}
\def\snd{\mathrm{snd}}

\def\refl{\mathrm{r}}
\def\d{\mathrm{d}}

\def\Cwf{\mathbf{Cwf}}

\def\dem{\mathrm{dem}}
\def\ctx{\mathrm{ctx}}

\def\Ucwf{\mathbf{Ucwf}}

\def\Law{\mathbf{Law}}
\def\Cartop{\mathbf{Cop}}
\def\Scwf{\mathbf{Scwf}}

\def\FL{\mathbf{FL}}

\def\LCC{\mathbf{LCC}}

\def\List{\mathrm{List}}

\newcommand{\RawCtx}{{\tt Ctx}}
\newcommand{\RawSub}{{\tt Sub}}
\newcommand{\RawTy}{{\tt Ty}}
\newcommand{\RawTm}{{\tt Tm}}

\def\vdashS{\vdash}
\def\emptySub{\langle \rangle}
\def\emptyContext{1}

\definecolor{Red}{rgb}{1,0,0}

\newcommand{\tuple}[1]{\langle #1 \rangle}
\newcommand{\natto}{\Rightarrow}
\newcommand{\intr}[1]{\llbracket #1 \rrbracket}
\def\CCs{\mathbf{CCs}}
\def\CCp{\mathbf{CCp}}
\def\CCCs{\mathbf{CCCs}}

\newcommand{\tto}{\Rightarrow}
\newdir{pb}{:(1,-1)@^{|-}}
\def\pb#1{\save[]+<16 pt,0 pt>:a(#1)\ar@{pb{}}[]\restore}
\newcommand{\Inv}{\mathsf{Inv}}

\begin{document}

\title*{Categories with Families: Unityped, Simply Typed, and Dependently Typed}
\author{Simon Castellan, Pierre Clairambault, and Peter Dybjer}
\institute{Simon Castellan \at Imperial College London, London SW7 2AZ, England\\ \email{simon.castellan@phis.me}
\and Pierre Clairambault \at CNRS, ENS de Lyon,
46 allée d'Italie,
69364 Lyon, France\\ \email{pierre.clairambault@ens-lyon.fr}
\and Peter Dybjer \at Chalmers University of Technology, SE-412 96 Göteborg, Sweden. \\\email{peterd@chalmers.se}}
%
%
\maketitle

\abstract*{Each chapter should be preceded by an abstract (no more than 200 words) that summarizes the content. The abstract will appear \textit{online} at \url{www.SpringerLink.com} and be available with unrestricted access. This allows unregistered users to read the abstract as a teaser for the complete chapter.\newline\indent
Please use the 'starred' version of the \texttt{abstract} command for typesetting the text of the online abstracts (cf. source file of this chapter template \texttt{abstract}) and include them with the source files of your manuscript. Use the plain \texttt{abstract} command if the abstract is also to appear in the printed version of the book.}

\abstract{We show how the categorical logic of untyped, simply typed and dependently typed lambda calculus 
can be structured around the notion of category with family (cwf). To this end we introduce subcategories of simply typed cwfs (scwfs), where types do not depend on variables, and unityped cwfs (ucwfs), where there is only one type. We prove several equivalence and biequivalence theorems between cwf-based notions and basic notions of categorical logic, such as cartesian operads, Lawvere theories, categories with finite products and limits, cartesian closed categories, and locally cartesian closed categories. Some of these theorems depend on the restrictions of contextuality (in the sense of Cartmell) or democracy (used by Clairambault and Dybjer for their biequivalence theorems). Some theorems are equivalences between notions with strict preservation of chosen structure. Others are biequivalences between notions where properties are only preserved up to isomorphism. In addition to this we discuss various constructions of initial ucwfs, scwfs, and cwfs with extra structure.}

\section{Introduction}

An important part of categorical logic is to establish correspondences between
languages (from logic) and categorical models. For example, in their book
\emph{Introduction to higher order categorical logic} \cite{LS86} Lambek and Scott prove
equivalences between typed lambda calculi and cartesian closed categories, between
untyped lambda calculi and C-monoids, and between intuitionistic type theories and
toposes. Lambek and Scott's intuitionistic type theories are intuitionistic versions
of Church's simple theory of types, which should not be confused with Martin-L\"of's
intuitionistic type theories. Interestingly, in the preface of their book \cite[p
viii]{LS86} Lambek and Scott express a desire to include a result concerning the
latter too:
\begin{quotation}
We also claim that intuitionistic type theories and toposes are closely related, in
as much as there is a pair of adjoint functors between their respective categories.
This is worked out out in Part II. The relationship between Martin-L\"of type
theories and locally cartesian closed categories was established too recently (by
Robert Seely) to be treated here.
\end{quotation}
Seely's seminal paper \cite{SeelyRAG:locccc} claims to prove that a category of Martin-L\"of type theories is equivalent to a category of locally cartesian closed categories (lcccs). However, his result relies on an interpretation of substitution as pullback, and the latter are only defined up to isomorphism. It is not clear how to choose pullbacks in such a way that the strict laws for substitution are satisfied. This coherence problem is identified and solved by Curien \cite{curien:fi} 
and Hofmann \cite{hofmann:csl} 
who provide alternative methods for interpreting Martin-L\"of type theory in lcccs (see also \cite{CurienGH14}). 
By using Hofmann's interpretation Clairambault and Dybjer \cite{ClairambaultD11,ClairambaultD14} show that there is an actual biequivalence of 2-categories. 

In this paper we ask ourselves what it would take to add the missing chapter on Martin-L\"of type theory and its correspondence with lcccs to the book by Lambek and Scott. 

First we would need to add some material to Part 0 in the book on ``Introduction to category theory", including introductions to lcccs, 2-categories, bicategories, pseudofunctors, and biequivalences. But more profoundly, our biequivalence theorem differs from Lambek and Scott's (and Seely's) equivalence theorems in important respects, since we replace Seely's category of Martin-L\"of theories by a 2-category of categories with families (cwfs), with extra structure for type formers $\Iext, \Sigma, \Pi$, and pseudo cwf-morphisms which preserve the structure up to isomorphism. Thus cwfs with extra structure replace Martin-L\"of theories on the ``syntactic" side of the biequivalence. 

This style of presenting the correspondence between ``syntax" and ``semantics" for Martin-L\"of's dependent type theory applies equally well to the simply typed lambda calculus and the untyped lambda calculus, provided we consider subcategories of simply typed cwfs (scwfs), where types do not depend on variables, and of unityped cwfs (ucwfs), where there is only one type. As for full cwfs we need to provide extra structure for modelling $\lambda$-abstraction and application in the untyped $\lambda$-calculus and for modelling type formers in the simply typed $\lambda$-calculus.

This suggests that we ought to rewrite Lambek and Scott's Part I ``Cartesian closed categories and $\lambda$-calculus" in a way which harmonizes with our presentation of the biequivalence between Martin-L\"of type theory and lcccs.

Categories wifh families (cwfs) model the most basic rules of dependent type theory,
those which deal with substitution, context formation, assumption, and general
reasoning about equality. A key feature of cwfs is that the definition can be
unfolded to yield a generalized algebraic theory in Cartmell's sense
\cite{cartmell:apal}. As such it suggests a language of cwf-combinators which can be
used for the construction of initial cwfs (with extra structure for modelling type
formers). 

We prove several correspondence theorems between ``syntax" in the guise of a number of cwf-based notions and ``semantics" in the guise of some basic notions from category theory. Some of our theorems require ``contextuality", a notion introduced by Cartmell \cite{cartmell:apal} for his contextual categories. Others require ``democracy", a notion introduced by Clairambault and Dybjer for their biequivalence theorems. Moreover, our equivalence theorems require strict preservation of chosen cwf-structure, while our biequivalence theorems only require preservation of cwf-structure up to isomorphism. 
In this way we can relate a number of notions from categorical logic such as cartesian operads, Lawvere's algebraic theories, Obtu\l owicz' algebraic theories of type $\lambda$-$\beta\eta$ \cite{Obtulowicz77}, categories with finite products and limits, cccs, and lcccs, to the corresponding cwf-based notions. In addition to this we discuss different constructions of initial ucwfs, scwfs and cwfs (with extra structure) with and without explicit substitutions. 

The purpose of our work is not so much to prove new results, but to suggest a new way
to organize basic correspondence theorems in categorical logic, where the
ucwf-scwf-cwf-sequence provides a smooth progression of the categorical model theory
of untyped, simply typed, and dependently typed $\lambda$-calculi. We will also
highlight some of the subtleties which arise when relating syntactic and semantic
notions. Another important feature is that the correspondences between logical
theories and categorical notions are now split into two phases: (i) equivalences and
biequivalences between cwf-based notions and basic categorical notions, and (ii) the
constructions of initial cwf-based notions. This yields an ``abstract syntax''
perspective of formal systems, where specific formalisms for untyped, simply typed
and dependently typed $\lambda$-calculi are instances of the respective isomorphism
classes of initial cwf-based notions. This is particularly important for dependent
types and Martin-L\"of type theory, since different authors make different choices in
the exact formulation of the syntax and inference rules. Being initial in the
appropriate category of cwfs is a suitable correctness criterion for these
formulations. 

In the text we have only discussed the relationship to some of the most
important related notions in the literature. We would like to emphasise that we
have not at all tried to give a comprehensive account of related work. Such an
account would be a daunting task, since we would need to cover a multitude of
works on models of the untyped and simply typed lambda calculus and of
dependent type theory. Nevertheless, we would like to refer to
Jacobs' book \emph{Categorical Logic and Type Theory}
\cite{jacobs1999categorical} which provides a comprehensive account of
categorical type theory as fibred category theory. In particular, Chapter 2 on
Simple Type Theory and Chapter 10 on First Order Dependent Type Theory contain
much related material. In Jacobs' work the notion of fibration takes centre
stage while its role is only implicit in our cwf-based approach. In any cwf we
can define an indexed category (and hence a fibration) of types indexed by contexts. However, an account of the precise correspondence between the two approaches is outside the scope of this paper.

\subsection*{Plan of the paper} 

In Section 2 we introduce cwfs and explain their connection to Martin-L\"of type
theory. We also define contextual and democratic cwfs. In Section 3 we consider
unityped cwfs and show that contextual ucwfs are equivalent both to cartesian operads
and to Lawvere theories. We then add extra structure to model the untyped
$\lambda\beta\eta$-calculus, and show how initial ucwfs can be built both as calculi
of explicit and implicit substitutions. In Section 4 we consider simply typed cwfs.
We first show the equivalence between contextual scwfs with finite product types and
categories with finite products as structure. Then we show the biequivalence between
democratic scwfs and categories with finite products as property.  Moreover, we add
function types to our contextual scwfs with finite product types and show their
equivalence to cartesian closed categories as structure. This is our analogue of
the equivalence between simply typed $\lambda$-calculi and cartesian closed
categories in Lambek and Scott. In Section 5 we discuss alternative definitions of
full dependently typed cwfs. We also show an explicit construction of a free cwf.
Moreover, we present two biequivalences. The first is between categories with finite limits
and democratic cwfs with extensional identity types and $\Sigma$-types. The other is
between lcccs and democratic cwfs with extensional identity types, $\Sigma$-types and
$\Pi$-types. Finally, we outline the construction of a free lccc, and show how we can use
our biequivalence theorem to prove that equality in this lccc is undecidable. In this
section we only give an overview of the theorems and refer the reader to the
journal articles \cite{ClairambaultD14,castellan:lmcs} for detailed proofs.

%


\section{Dependent type theory and categories with families}

We recall the general structure of judgments and inference rules of Martin-L\"of type
theory and explain its connection to the definition of cwfs.

\subsection{Martin-L\"of type theory}

We here consider Martin-L\"of type theory with extensional identity types in
the style of \cite{martinlof:hannover,martinlof:padova}. 

\subsubsection{Judgments}

In Martin-L\"of \cite{martinlof:hannover} intuitionistic type
theory is presented as a formal system with four forms of judgment: 
\begin{eqnarray*}
\Gamma &\vdash& A\ \mathtt{ type}\\
\Gamma &\vdash& A = A'\\
\Gamma &\vdash& a : A \\
\Gamma &\vdash& a = a' : A 
\end{eqnarray*}
These respectively state that $A$ is a well-formed type; that $A$ and
$A'$ are equal types; that $a$ is a valid term of type $A$; and that $a$ and
$a'$ are equal terms of type $A$. 
These four forms of judgments are hypothetical, that is, relative to a context $\Gamma
= x_1 : A_1, \ldots, x_n : A_n$ which assigns types $A_i$ to free variables $x_i$.  

Martin-L\"of \cite{martinlof:gbg92} presents an alternative version of the
theory in the form of a {\em substitution calculus} (see also
\cite{tasistro:lic}) with four additional forms of judgment:
\begin{eqnarray*}
&&\Gamma\ \mathtt{ context}\\
&&\Gamma = \Gamma'\\
&&\Delta \to \gamma : \Gamma \\
&&\Delta \to \gamma = \gamma' : \Gamma
\end{eqnarray*}
One aim is to make the rules for context formation explicit. Another is to
formulate a calculus where substitution is a first-class citizen (a
term constructor) and not just an operation defined by induction on the types
and terms of the theory. 
Given $\Gamma = x_1: A_1, \dots, x_n: A_n$, the judgment $\Delta \to
\gamma : \Gamma$ expresses that $\gamma$ is a substitution (an assignment) of
terms for free variables $x_1 = a_1, \ldots, x_n = a_n$, where $a_1 : A_1, \ldots, a_n : A_n$ are terms in context
$\Delta$. Substitutions can be applied to both terms and types, \emph{e.g.}, if
$\Delta \to \gamma : \Gamma$ and $\Gamma \vdash A\ \mathtt{ type}$ and $\Gamma
\vdash a : A$, then $\Delta \vdash A[\gamma]\ \mathtt{ type}$ and $\Delta \vdash a[\gamma] :
A[\gamma]$. But rather than defining $a[\gamma]$ and $A[\gamma]$ (which simultaneously substitute terms for free variables) by induction, they are
instead explicit term constructors, and the effect of replacing a variable $x_i$
by a term $a_i$ is expressed \emph{a posteriori} via judgmental equality.

\subsubsection{Inference rules}\label{subsubsec:inferencerules}

The inference rules of intuitionistic type theory can be separated into two
kinds. 
The first kind are the general rules, the most basic rules for reasoning
with dependent types. They deal with \emph{substitution}, \emph{context
formation}, \emph{assumption}, and \emph{general equality reasoning}. They form
the backbone of the dependently typed structure, and carry no information yet about specific
term and type formers.
The second kind consists of the rules for type formers,
such as $\Pi, \Sigma,$ and identity types. These rules are divided into \emph{formation}, \emph{introduction}, \emph{elimination}, and
\emph{equality rules} (also called \emph{computation rules}).

\emph{Categories with families} capture models of the
first kind of rules: the backbone of Martin-L\"of type theory
which is independent of type and term constructors.

\subsection{Categories with families}
 

We first give the definition and then explain the connection to type theory.

\subsubsection{Definition}
\label{subsubsec:def_cwf}
The definition uses the category $\Fam$ of \emph{families
of sets}. Its objects are families $(U_x)_{x\in X}$. A morphism
with source $(U_x)_{x\in X}$ and target $(V_y)_{y\in Y}$ is a pair consisting
of a reindexing function $f : X \to Y$, and a family $(g_x)_{x\in X}$ where for
each $x \in X$, $g_x : U_x \to V_{f(x)}$ is a function.


\begin{definition}\label{def:cwf}
A category with families (cwf) consists of the following:
\begin{itemize} 
\item A category $\C$ with a terminal object $1$.

\emph{Notation and terminology.} We use $\Gamma, \Delta,$ \emph{etc}, to range over objects of $\C$, and refer to
those as {\em contexts}. Likewise, we use $\delta, \gamma,$ \emph{etc}, to range
over morphisms, and refer to those as {\em substitutions}. We refer to $1$ as the
\emph{empty context}. We write $\tuple{}_\Gamma \in \C(\Gamma, 1)$ for the
terminal map, representing the {\em empty substitution}.
\item A $\Fam$-valued presheaf, \emph{i.e.} a functor
$T : \Cop \to \Fam$.

\emph{Notation and terminology.} If $T(\Gamma) = (U_x)_{x\in X}$, we write $X = \Ty(\Gamma)$ and refer to
its elements as \emph{types in context $\Gamma$} -- we use $A, B, C$ to range over
such {\em types}. For $A \in X = \Ty(\Gamma)$, we write $U_A = \Tm(\Gamma, A)$
and refer to its elements as \emph{terms of type $A$ in context $\Gamma$.} 
Finally, for $\gamma : \Delta \to \Gamma$, the functorial action yields
\[
T(\gamma) : (\Tm(\Gamma, A))_{A\in \Ty(\Gamma)} \to 
		(\Tm(\Delta, B))_{B\in \Ty(\Delta)}
\]
consisting of a pair of a reindexing function $\_\,[\gamma] : \Ty(\Gamma) \to
\Ty(\Delta)$ referred to as \emph{substitution in types}, and for each $A\in
\Ty(\Gamma)$ a function $\_\,[\gamma] : \Tm(\Gamma, A) \to \Tm(\Delta,
A[\gamma])$ referred to as \emph{substitution in terms}.
\item A \emph{context comprehension operation} which to a given context $\Gamma
\in \Cobj$ and type $A \in \Ty(\Gamma)$ assigns a context $\Gamma \cext
A$ and two projections
\[
\p_{\Gamma, A} : \Gamma \cext A \to \Gamma
\qquad\qquad
\q_{\Gamma, A} \in \Tm(\Gamma\cext A, A[\p_{\Gamma,A}])
\] 
satisfying the following universal property: for all $\gamma : \Delta \to
\Gamma$, for all $a\in \Tm(\Delta, A[\gamma])$, there is a unique
$\tuple{\gamma, a} : \Delta \to \Gamma \cext A$ such that
\[
\p_{\Gamma, A} \circ \tuple{\gamma, a} = \gamma
\qquad \qquad
\q_{\Gamma, A} [\tuple{\gamma, a}] = a\,.
\]
%
We say that $(\Gamma\cext A, \p_{\Gamma, A}, \q_{\Gamma, A})$ is a
\emph{context comprehension} of $\Gamma$ and $A$.
\end{itemize}
\end{definition}

Observe the similarity between the universal properties of context comprehension and cartesian products -- the former is a skewed dependently typed version of the latter. It is also closely related to Lawvere comprehension \cite{lawvere:hyperdoctrines}. 

This definition is the standard, historical definition of cwfs
\cite{dybjer:torino}. As notations and terminology suggest, it is closely
connected to the syntax of type theory, and particularly to Martin-L\"of's
\emph{substitution calculus}.

\begin{remark}
The structure from Definition \ref{def:cwf} exactly matches that of
Martin-L\"of's \emph{substitution calculus} mentioned before. The correspondence
follows:
\begin{itemize}
\item $\Gamma \in \Cobj$ models the judgment $\Gamma\ \mathtt{ context}$ and
$\Gamma = \Gamma' \in \Cobj$ models $\Gamma = \Gamma'$.
\item $\gamma \in \C(\Delta,\Gamma)$ models the judgment $\Delta \to \gamma :
\Gamma$ and
$\gamma = \gamma' \in \C(\Delta,\Gamma)$ models $\Delta \to \gamma = \gamma' :
\Gamma$.
\item
$A \in \Ty(\Gamma)$ models the judgment $\Gamma \vdash A\ \mathtt{ type}$ and
$A = A' \in \Ty(\Gamma)$ models $\Gamma \vdash A = A'$.
\item
$a \in \Tm(\Gamma,A)$ models the judgment $\Gamma \vdash a : A$ and
$a = a' \in \Tm(\Gamma,A)$ models $\Gamma \vdash a = a' : A$.
\end{itemize}
\end{remark}

The connection with Martin-L\"of's substitution calculus contributes to the appeal of cwfs: they give rise to categorical combinators for dependent types just as cccs give rise to categorical combinators for the simply typed $\lambda$-calculus \cite{Curien86}.
However, Definition \ref{def:cwf} has sometimes been criticized for being
\emph{too} close to the syntax, or for relying on less standard mathematical
objects such as $\Fam$-valued presheaves. In Section
\ref{subsec:plaincwfs} we will discuss alternative formulations of cwfs highlighting other
aspects of the structure.


%


\begin{remark}
A key feature of the notion of cwf is that it can be presented as a generalized algebraic theory in the sense of Cartmell \cite{cartmell:apal}.\begin{itemize}
\item The generalized algebraic theory of categories introduces the sorts
$\Cobj$ and $\C(\Delta,\Gamma)$, the operations $\gamma \circ \delta$ and
$\id_\Gamma$, and associativity and identity laws. 
\item The $\Fam$-valued presheaf adds the sorts $\Ty(\Gamma)$ and
$\Tm(\Gamma, A)$, the operations $A[\gamma]$ and $a[\gamma]$, and associativity and identity laws for both.
\item The terminal object adds the operation 1 and $\tuple{}$, and its uniqueness law.
\item Context comprehension adds the operations 
$\Gamma \cext A, \p_{\Gamma,A}, \q_{\Gamma,A}$, and $\tuple{\gamma , a}$, and the projection and surjective pairing laws.
\end{itemize}
See \cite{dybjer:torino} for a complete presentation of the generalized algebraic theory of cwfs.
We remark that we have suppressed some of the arguments of the operations. For example, composition is officially an operation with five arguments: $\Xi, \Delta, \Gamma \in \Cobj,\delta \in \C(\Xi,\Delta)$, and $\gamma \in \C(\Delta,\Gamma)$, but we suppress the three first when we write $\delta \circ \gamma$. Similar remarks hold for the operations $A[\gamma], a[\gamma],\tuple{}$, and $\tuple{\gamma , a}$. We sometimes drop even more arguments to simplify notation and for example write $\p_A$ or $\p$ for the official $\p_{\Gamma,A}$, etc. Moreover, we sometimes write $\gamma : \Delta \to \Gamma$ for $\gamma \in \C(\Delta,\Gamma)$.

A cwf is thus a structure $(\C,T,1,\tuple{},\cext , \p, \q, \langle -,- \rangle)$, subject to some equations. However, we often refer to a cwf by the first 
two components $(\C,T)$ or even the first component $\C$. 
\end{remark}

As already mentioned, cwfs only organize the core of dependent type
theory, the basic structure and operations on contexts, types, terms, and substitutions.
We will see later how cwfs naturally generalize well-known notions in categorical logic to dependent types. We will also see how they
may be enriched with type and term formers, in order to capture Martin-L\"of type theory with $\Sigma$-types, $\Pi$-types
and identity types. Finally, we will see that the syntax of Martin-L\"of type
theory may be defined as the \emph{initial cwf} in a precise sense.

%
%
%

\subsubsection{Structure of contexts} The definition of cwfs just contains two operations on contexts: the terminal object representing the empty context and an operation mapping a
context $\Gamma$ and a type $A \in \Ty(\Gamma)$ to a new context
$\Gamma\cext A$. It is however not required that all contexts are generated by repeated application of these two rules. In contrast to this, Cartmell \cite{cartmell:apal} adds such a constraint on the structure of context for his contextual categories. We shall use the following formulation, which is equivalent to Cartmell's:
\begin{definition}[Contextuality]
A cwf is {\em contextual} iff there is
$$l : \Cobj \to \N\,,$$
a \emph{length function}, 
such that $l(\Gamma) = 0$ iff $\Gamma = 1$, and $l(\Gamma) = n+1$ iff there are
unique $\Delta \in \Cobj$ and $A\in \Ty(\Delta)$ such that $\Gamma =
\Delta\cext A$, and $l(\Delta) = n$.
\end{definition}
Although this requirement will be used in some of our equivalence theorems, it is not part of our definition of cwf. The reason is that unlike the other parts of the definition of cwfs, 
it does not correspond to an inference rule of dependent type theory, and
it is not expressed in the language of generalized algebraic theories.
However, the free cwf is contextual.

Without going as far as requiring all contexts to be defined
inductively, we sometimes wish to overcome the intrinsic distinction between contexts
and types by asking that up to isomorphism, every context is represented by
a type \cite{ClairambaultD11,ClairambaultD14}.

\begin{definition}[Democracy]
A cwf is {\em democratic} provided each context $\Gamma$ is represented by a type $\overline{\Gamma}$ in the sense that there is an isomorphism 
$$\gamma_\Gamma : \Gamma  \cong 1.\overline{\Gamma}$$
\end{definition}

Democracy does not imply contextuality. However, 
in the presence of unit types and $\Sigma$-types, the converse
holds: any context $1\cdot A_1 \cdot \ldots \cdot A_n$ may be represented by the
iterated $\Sigma$-type $\Sigma(A_1, \Sigma(A_2, \dots, \Sigma(A_{n-1},
A_n)\cdots)) \in \Ty(1)$. Like contextuality, democracy does not correspond to
an inference rule of dependent type theory. However, unlike
contextuality, democracy can be expressed in the language of
generalized algebraic theories.  

\subsubsection{Strict morphisms of cwfs} 
We will now introduce a notion of morphisms between cwfs.

\begin{definition}
A \emph{(strict) cwf-morphism between cwfs $(\C,T_\C)$ and $(\D,T_\D)$} is
a pair $(F, \sigma)$ where $F : \C \to \D$
is a functor preserving $1$ on the nose, and
\[
\sigma : T_\C \natto T_\D \circ F
\]
is a natural transformation between $\Fam$-valued presheaves,
preserving context comprehension on the nose.
\end{definition}

Thus we have $$\sigma_\Gamma :
\Ty_\C(\Gamma) \to \Ty_\D(F \Gamma)$$
\[
\sigma_\Gamma^A : \Tm_\C(\Gamma, A) \to \Tm_\D(F\Gamma, \sigma_\Gamma (A))
\]
for $\Gamma \in \Cobj$ and $A\in
\Ty(\Gamma)$.

It is convenient to simplify notations and write all the components
of a cwf-morphism as $F$ so that we write $F(A)$ for $\sigma_\Gamma(A)$ and $F(a)$ for $\sigma^A_\Gamma(a)$. Naturality of $\sigma$
amounts to preservation of substitution, \emph{i.e.}, for all $\gamma
: \Delta \to \Gamma$ in $\C$, we have
\[
F(A[\gamma]) = (FA)[F\gamma] \qquad \qquad F(a[\gamma]) =
(Fa)[F\gamma]\,.
\]

Finally, preservation of context comprehension on
the nose means that $F(\Gamma\cext A) = F\Gamma\cext
FA$, with $F(\p_{\Gamma, A}) = \p_{F\Gamma, FA}$ and $F(\q_{\Gamma,
A}) = \q_{F\Gamma, FA}$.

Small cwfs and strict cwfs-morphisms form a category, written $\Cwf$. 

\section{Unityped cwfs}

As we explained in the introduction, a key feature of the notion of cwf is that it
can be presented as a generalized algebraic theory. As a consequence it can be seen
both as a notion of model and as an (idealized) language for dependent type theory.
It is therefore a suitable intermediary between traditional formal systems for
dependent type theory and categorical notions of models. This paper is based on the
observation that restricted classes of cwfs can play a similar role for untyped and
simply typed systems. 

In this section we will look at cwfs with only one type and
claim that they play a similar role for untyped systems as cwfs do for dependently
typed systems.

%

\subsection{Plain ucwfs}

The definition of cwfs with only one type can be simplified as follows.

\begin{definition}\label{def:Ucwf}
A unityped category with families (ucwf) consists of the following:
\begin{itemize}
\item A category $\C$ with a terminal object, written $0$.

\emph{Notation and terminology.} We use $n, m,$ \emph{etc} to range over objects of $\C$,
and refer to those as {\em contexts}. Likewise, we use $\delta,
\gamma,$ \emph{etc} to range over morphisms, and refer to those as
{\em substitutions}. We refer to $0$ as the \emph{empty context}. We
write $\tuple{}_n \in \C(n, 0)$ for the terminal morphism,
representing the empty substitution.
\item A presheaf $\Tm : \Cop \to \Set$.

\emph{Notation and terminlogy.} We refer to the elements of $\Tm(n)$ as the \emph{terms of arity $n$}
-- we use $a, b,$ etc to range over terms. Finally, for $\gamma : n
\to m$, the functorial action of $\Tm$ yields a \emph{substitution
operation}
\[
\Tm(\gamma) = \_\,[\gamma] : \Tm(m) \to \Tm(n)\,.
\]
\item A \emph{context comprehension operation} which to a given
context $n \in \Cobj$ assigns a context $s(n)\in \Cobj$ along with two
\emph{projections}
\[
\p_n : s(n) \to n
\qquad\qquad
\q_n \in \Tm(s(n))
\]
satisfying the following universal property: for all $\gamma : m
\to
n$, for all $a\in \Tm(m)$, there is a unique
$\tuple{\gamma, a} : m \to s(n)$ such that
\[
\p_n \circ \tuple{\gamma, a} = \gamma
\qquad \qquad
\q_n [\tuple{\gamma, a}] = a\,.
\]
\end{itemize}
\end{definition}

\begin{remark}
The context comprehension operation for ucwfs amounts to the assignment of a representation
$s(n)$ of the presheaf
$$
\C(-,n) \times \Tm(-) : \C^\op \to \Set
$$
for all $n \in \Cobj$.
\end{remark}

If $n \in \mathbb{N}$ is a natural number, we write $\underline{n} =
s^n(0)$ for the context obtained by $n$ applications of the context
comprehension operation. If the ucwf is contextual, all objects of $\C$
have the form $\underline{n}$ for $n \in \mathbb{N}$.
We think of the terms in $\Tm(\underline{n})$ as \emph{boxes} with
$n$ inputs and one output. Substitutions $\gamma : \underline{m}
\to \underline{n}$ are boxes with $m$ inputs and $n$ outputs. We have
\[
\gamma = \tuple{\tuple{\dots\tuple{\tuple{}, a_1},\dots},a_n} :
\underline{m} \to \underline{n}
\]
for $a_1, \ldots, a_n \in \Tm(\underline{m})$. For convenience
we write $\gamma = \tuple{a_1, \dots, a_n}$. For $a \in \Tm(\underline{m})$,
performing the substitution $a[\gamma]$ amounts to connecting the
box $a_i$ to the $i$-th input of the box $a$, see
Figure \ref{fig:boxes}.

\begin{figure}
\begin{center}
\[
\includegraphics[scale=1]{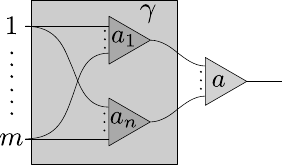}
\]
\end{center}
\caption{Substitution as plugging boxes}
\label{fig:boxes}
\end{figure}

Figure \ref{fig:boxes} reminds us of other categorical notions that aim to capture
algebraic theories, such as Lawvere theories. We will come back
later to this similarity. However, Figure \ref{fig:boxes} is
misleading in one respect: it suggests that the boxes $a_1, \dots,
a_n$ have free variables in $\Fin(m) = \{1, \dots, m\}$. In reality, in the
context of ucwfs, these free variables are not first-class citizens
but are obtained indirectly through sequences of \emph{projections}.
More precisely, the term
\[
\pi_i^m = \q_{i-1} [ \p_{i} ] \cdots [ \p_{m-1} ] \in \Tm(m)
\]
will serve in place of the free variable $i$ in the context of size
$m$. With the notations introduced, we then have the expected equation
\[
\pi_i^m [\tuple{a_1, \dots, a_{m}}] = a_i
\]

These notations suggest a
correspondence to \emph{cartesian
operads} \cite{trimble:nlab}, and we will come back to this
connection.
Before that, writing $\Ucwf$ for the category of small ucwfs and strict
cwf-morphisms, we note:

\begin{proposition}\label{prop:ucwf_initial}
The category $\Ucwf$ has an initial object.
\end{proposition}
\emph{Construction 1.} From the definition of ucwf it is immediately clear
that an initial ucwf $\T_{\mathrm{ucwf}}$
can be
generated inductively:
we simultaneously define the three families $\Cobj, \C(n,m)$, and
$\Tm(n)$ where the ucwf-operations become introduction rules. Then we take the quotients with respect to the equivalence relations generated by the ucwf-equations.

\emph{Construction 2.} Alternatively, we can construct it as the presheaf of variables over the category of renamings.
\begin{itemize}
\item The category $\Ncat$ of \emph{renamings}, with
objects $\Ncat_0 = \N$ and $\Ncat(n,m) = \Fin(n)^m$.
\item The \emph{presheaf $\Tm : \Ncat^{\mathrm{op}} \to \Set$} is defined by
$\Tm(n) = \Fin(n)$ and $i[(a_1,\ldots,a_{m})] = a_i$.
\item \emph{Context comprehension} is defined by $s(n) = n+1$, $\p_n =
(1, \dots, n)$, and $\q_n = n+1$.
\end{itemize}
This construction and its isomorphism with
$\T_{\mathrm{ucwf}}$ have been formalized in Agda by Brilakis
\cite{Brilakis18}.

Interestingly, this is the same as the free cartesian operad. This
initial ucwf is
\emph{contextual}. As we will see, this requirement is necessary for
the connection with cartesian operads.

\subsection{Contextual ucwfs}
\label{subsec:contextual_ucwfs}

\subsubsection{Cartesian operads}
Let us now consider the special case of a \emph{contextual ucwf} $\C$, where
the length function induces a bijection $\C_0 \cong \N$. 
It follows from the laws of cwfs that
\[
\C(\underline{m},\underline{n}) \cong \Tm(\underline{m})^n
\]
From right to left we use the $n$-ary tupling
introduced above, while from left-to-right we apply projections. This is essentially
the same data as for \emph{cartesian operads}.

\begin{definition}
A {\em cartesian operad} consists of the following:
\begin{itemize}
\item a family $\Tm(n)$, where $n \in \N$, of {\em n-ary operations};
\item an operation of {\em operad composition} which maps $a \in \Tm(m)$ and $\gamma \in \Tm(n)^m$ to
$a[\gamma] \in \Tm(n)$ and satisfies identity and
associativity laws;
\item {\em projections} $\pi^n_i \in \Tm(n)$, such that
$
\pi_i^n[(a_1, \ldots, a_{n})] = a_i\,.
$
\end{itemize}
\end{definition}
Focusing on contextual ucwfs allows us to extract the mechanism for terms and
substitutions that is at play in full cwfs, but without considering types. 

Writing $\Ucwf_\ctx$ of $\Ucwf$ having as objects \emph{contextual} ucwfs, and $\Cartop$
for the category of \emph{cartesian operads}, we have:

\begin{theorem}
The categories $\Ucwf_\ctx$ and $\Cartop$ are equivalent.
\end{theorem}

Rather than formally define $\Cartop$ and prove this equivalence, we shall detail it for
\emph{Lawvere theories}, which form a category equivalent to
$\Cartop$ \cite{trimble:nlab}.

\subsubsection{Lawvere theories}

Contextual ucwfs are also equivalent to \emph{Lawvere theories}.
The following definition is from
\cite{HylandP07}.

\begin{definition}
A \emph{Lawvere theory} consists of a small category $\C$ with (necessarily
strictly associative) finite products and a strict finite-product preserving
identity-on-objects functor $L : \Ncat \to \C$. Here, $\Ncat$ refers to the category of \emph{renamings}
introduced in the proof of Proposition \ref{prop:ucwf_initial}, with terminal object $0$ and binary products defined by $+$.

A map of Lawvere theories from $(L, \C)$ to $(L', \C')$ is a 
(necessarily strict) finite-product preserving functor from $\C$ to $\C'$
that commutes with the functors $L$ and $L'$. Lawvere theories and their maps
form a category $\Law$.
\end{definition}

Note that $L : \Ncat \to \C$ is usually (equivalently) presented as 
$\Ncatop_0 \to \C$, where $\Ncat_0$, a skeleton of the category of finite sets
and functions, has an initial object $0$ and finite coproducts given by 
the sum of integers. 


\begin{theorem}
The categories $\Ucwf_\ctx$ and $\Law$ are equivalent.
\end{theorem}
\begin{proof}
Let $(\C, T_\C)$ be a contextual ucwf. We already noted that $\C_0
\cong \mathbb{N}$. Moreover, we observe that there is a \emph{unique} contextual
ucwf $(\D, T_\D)$, isomorphic to $(\C, T_\C)$ and such that $\D_0 = \mathbb{N}$
with $0$ terminal and $s(n) = n+1$ for all $n \in \D_0$. Hence we assume from
now on that contextual ucwfs have natural numbers as objects.
By Proposition \ref{prop:ucwf_initial} below, $\Ncat$ is the base
category of the \emph{initial ucwf}. Hence, if  $(\C, T_\C)$ is a contextual ucwf, there is a unique functor 
\[
L : \Ncat \to \C,
\]
which is the first component of a cwf-morphism. In particular, 
$L$ preserves the terminal object and context comprehension on the nose. It follows that it is identity-on-objects and strictly
finite-product preserving.  

If $(F, \sigma)$ is a cwf-functor between contextual ucwfs (assuming
\emph{w.l.o.g.} that these ucwfs have $\mathbb{N}$ as objects), it follows that $F$ is a morphism between the corresponding Lawvere theories.
Conversely, for any morphism $F$ between the corresponding Lawvere theories,
$\sigma$ can be uniquely recovered from $F$ and projection.
This yields a full and faithful functor from $\Ucwf_\ctx$ to $\Law$. 

Finally this functor is surjective on objects: from $(L, \C)$ a Lawvere theory,
there is a ucwf with category $\C$; terms $\Tm(n) = \C(n, 1)$, and context
comprehension $s(n) = n + 1$. The universal property follows from that of the
finite product. From all that, it follows that $\Ucwf_\ctx$ and $\Law$ are
equivalent.
\end{proof}

\begin{remark}
In a recent paper Fiore and Voevodsky \cite{VoevodskyF17} prove a closely related result about C-systems, a variant of Cartmell's contextual categories. They prove that their category of Lawvere theories is isomorphic to the subcategory of C-systems whose length functions (in the definition of contextuality) are bijections. 

\end{remark}

\subsection{$\lambda\beta\eta$-ucwfs}

Ucwfs give rise to a generalized algebraic theory which captures the
combinatorics of terms and substitution in a similar way as
cartesian operads and Lawvere theories. This primitive structure may
then be enriched with operations and equations for capturing specific theories, such as the pure $\lambda\beta\eta$-calculus.

\begin{definition}
A $\lambda\beta\eta$-ucwf is a ucwf $(\C,\Tm)$ with two more operations:
\begin{eqnarray*}
\lambda_n &:& \Tm(s(n)) \to \Tm(n)\\
\ap_n &:& \Tm(n) \times\Tm(n) \to \Tm(n)
\end{eqnarray*}
for all $n \in \C_0$, and four more equations:
\[
\begin{array}{rclcl}
\lambda_n(b)[\gamma] &=& \lambda_{m}(b[\tuple{\gamma \circ
\p_m,\q_m}]) \\
\ap_n(c,a)[\gamma]&=& \ap_m(c[\gamma],a[\gamma])\\
\ap_n(\lambda_n(b),a) &=& b[\tuple{\id_n,a}] &\qquad&(\beta)\\
\lambda_n(\ap_{s(n)}(c[\p_n],\q_n)) &=& c&&(\eta)
\end{array}
\]
for $\gamma : m \to n$, $b \in \Tm(s(n))$, and $c, a \in \Tm(n)$. A
\emph{$\lambda\beta$-ucwf} has the same operations, but it not
subject to the $(\eta)$ equation.
\end{definition}

The definition above is natural and
close to the syntax. As we will see later on, it is the direct
simplification of the notions of arrow and $\Pi$-types in the
simply-typed and the dependently typed case discussed later on.
However, in the unityped case, this definition can be simplified dramatically.

\begin{proposition}
Let $(\C, \Tm)$ be a ucwf. Then, $\lambda\beta\eta$-structures on
$\C$ are equivalently defined as natural isomorphisms between presheaves
\[
\Tm(\s(-)) \stackrel{\lambda}{\cong} \Tm(-)
\]
where the functorial action of $\s$ is defined as $\s(\gamma) =
\tuple{\gamma \circ \p_m, \q_m}$. More precisely, \emph{(1)} for any
$\lambda\beta\eta$-structure $\lambda$ is such a natural
isomorphism, and \emph{(2)} given such a natural isomorphism, there
is a unique $\lambda\beta\eta$-structure giving rise to it.
\end{proposition}
\begin{proof}
For \emph{(1)}, given a $\lambda\beta\eta$-structure on $(\C, \Tm)$,
we first observe that $\lambda_n$ is natural by the substitution law.
For $c \in \Tm(n)$ we set $\lambda_n^{-1}(c) =
\ap_{\s(n)}(c[\p_n],\q_n) \in \Tm(\s(n))$ -- using $\beta, \eta$ and
the substitution law for $\lambda$,  $\lambda_n$ and
$\lambda_n^{-1}$ are inverse. 

For \emph{(2)}, given a natural iso $\lambda_n$, we set $\ap_n(c,
a) = \lambda_n^{-1}(c)[\tuple{\id_n, a}]$. The $\beta$-rule follows from the fact that $\lambda_n^{-1}\circ
\lambda_n$ is the identity. The $\eta$-rule follows from the
naturality of $\lambda^{-1}$ plus the fact that $\lambda_n \circ
\lambda_n^{-1}$ is the identity. The substitution law for $\lambda$
is by naturality of $\lambda$, and the substitution law for $\ap$ is
by naturality of $\lambda^{-1}$. Finally, uniqueness of the
$\lambda\beta\eta$-structure (\emph{i.e.} of $\ap_n$) relies on the
substitution rule for $\ap$.
\end{proof}

%

\subsubsection{Some related models of the untyped $\lambda$-calculus}

There are many notions of model of $\lambda$-calculus, see for example Barendregt \cite{barendregt:lambda}. We will only briefly discuss the ones given by 
Obtu\l owicz \cite{Obtulowicz77}, Aczel \cite{aczel:frege}, and Lambek and Scott \cite{LS86}.

Obtu\l owicz's \emph{algebraic theories of type $\lambda$--$\beta\eta$} are Lawvere theories similar to the Lawvere theories corresponding to contextual 
$\lambda\beta\eta$-ucwfs, but 
use an evaluation morphism $\ev$ as a primitive instead of $\ap$.
These operations are interdefinable, via $\ev = \ap(\pi^2_1,\pi^2_2)
\in \Tm(2)$ and $\ap(c,a) = \ev [\tuple{ c , a } ]\in \Tm(n)$ for
$c,a \in \Tm(n)$. 

As a basis for his notion of Frege structure, Aczel introduces a notion of \emph{lambda structure}. This in turn is based on
the notion of an {\em explicitly closed family}, which is a cartesian operad where $\Tm(n) \subseteq \Tm(0)^n \to
\Tm(0)$, so that $a[\gamma]$ is function composition, and projections are the
projections in the metalanguage. It is thus a cartesian operad which is well-pointed in
the sense that $a, a' \in \Tm(n)$ and $a [ \gamma ] = a' [ \gamma ]$ for all
$\gamma \in \C(0,n)$ implies $a = a'$. To model the $\lambda\beta$-calculus Aczel adds two operations
\begin{eqnarray*}
\lambda_0 &:& \Tm(1) \to \Tm(0)\\
\ap_0 &:& \Tm(0) \times\Tm(0) \to \Tm(0)
\end{eqnarray*} 
The resulting notion of lambda structure is equivalent to well-pointed
$\lambda\beta$-ucwfs. Since terms are
\emph{functions}, there is a unique way to define the operations
$\lambda_n$ and $\ap_n$ for $n>0$ so that they satisfy the
substitution laws of $\lambda\beta$-ucwfs:
\begin{eqnarray*}
\lambda_n(b) (\gamma) &=& \lambda_0(b(\gamma \circ \p_0,\q_0))\\
\ap_n(c,a) (\gamma) &=& \ap_0(c(\gamma),a(\gamma))
\end{eqnarray*}
for $\gamma \in \C(0,n)$. The general substitution rules follow from
this definition.

Lambek and Scott propose \emph{C-monoids} as their notion of model of the untyped $\lambda$-calculus.
These are monoids with extra structure coming from combinators of cartesian closed categories. C-monoids capture the equational behaviour of \emph{closed} rather than open terms. Like in
$\lambda\beta\eta$-ucwfs and cartesian closed categories, variables are dealt
with indirectly as projections. But in $\lambda\beta\eta$-ucwfs, variable
addressing is \emph{external}, that is, handled by the ucwf structure. There
are no \emph{term} constructors for pairs and projections -- in particular
closed terms, \emph{i.e.}, terms in $\Tm(0)$, do not form a C-monoid as they
support no pairing and projection operations. In contrast, C-monoids handle
variable addressing through pairs and projections at the \emph{term level}. 

We expect a strong relationship between C-monoids and  $\lambda\beta\eta$-ucwfs
with term-level pairs and projections. The proof should follow
\cite{LS86}, encoding open terms and substitution
within C-monoids via \emph{functional completeness}.  

%

\subsubsection{Initial $\lambda\beta\eta$-ucwfs} 

To conclude the discussion about
$\lambda\beta\eta$-ucwfs, we include a construction of the untyped
$\lambda$-calculus as the \emph{initial} such structure. For that, let us say
that a strict cwf-morphism $F$ between $\lambda\beta\eta$-ucwfs $(\C, \Tm_\C)$
and $(\D, \Tm_\D)$ is a \emph{strict $\lambda\beta\eta$-ucwf-morphism} iff the
action of $F$ on terms preserves all the term constructors on the nose.
Let us write $\Ucwf^{\lambda\beta\eta}$ for the category of small
$\lambda\beta\eta$-ucwfs and strict $\lambda\beta\eta$-ucwf-morphisms. Then, we
have:

\begin{proposition}
The category $\Ucwf^{\lambda\beta\eta}$ has an initial object.
\end{proposition}
\emph{Construction 1.}
The most direct method is similar to Construction 1 of an initial ucwf. We simultaneously define the three families $\Cobj, \C(n,m)$, and
$\Tm(n)$ where the $\lambda\beta\eta$-ucwf-operations become introduction rules. Then we take the quotients with respect to the equivalence relations generated by the $\lambda\beta\eta$-ucwf-equations.
This construction can be viewed as a well-scoped variable free version of the
$\lambda\sigma$-calculus of Abadi, Cardelli, Curien, and L\'evy \cite{AbadiCCL90}.

\emph{Construction 2.}
Another initial $\lambda\beta\eta$-ucwf can be constructed from the (well-scoped) $\lambda\beta\eta$-calculus. We let $\Cobj = \N$ and generate the family $\Tm(n)$ by the following rules
\begin{eqnarray*}
\var_n(i)&: & \Tm(n)\ \ \ (i \in \Fin(n))\\
\lambda_n &:& \Tm(s(n)) \to \Tm(n)\\
\ap_n &:& \Tm(n) \times\Tm(n) \to \Tm(n)
\end{eqnarray*}
quotiented by the equivalence relation $\sim$ generated by $\beta$ and $\eta$:
\begin{eqnarray*}
\ap_n(\lambda_n(b),a) &\sim& b[\tuple{\id_n,a}]\\
\lambda_n(\ap_{s(n)}(c[\p_n],\q_n)) &\sim& c
\end{eqnarray*}

Note that the variables $\var_n(i)$ where $i \in \Fin(n)$ were simply
represented by the number $i$ in the corresponding construction in Proposition
\ref{prop:ucwf_initial}.

We let $\C(n,m) = \Tm(n)^m$ and define substitution by induction on $\Tm(n)$:
\begin{eqnarray*}
\var_n(i) [(a_1,..., a_n)]&=& a_i\\
\lambda_n(b)[\gamma] &=& \lambda_{s(n)}(b[\tuple{\gamma \circ \p_m,\q_m}]) \\
\ap_n(c,a)[\gamma]&=& \ap_m(c[\gamma],a[\gamma])
\end{eqnarray*}
for $\gamma \in \C(m,n)$.

Note that this construction is an extension of
Construction 2 of an initial plain ucwf.

The two constructions, and the fact that both give rise to initial
$\lambda\beta\eta$-ucwfs, have been formalised in Agda by
Brilakis \cite{Brilakis18}. 
%
%
%
%

\section{Simply-typed cwfs}
\label{sec:scwfs}

\emph{En route} to full cwfs, we now add types
yielding \emph{simply-typed cwfs (scwfs)} (called \emph{non-dependent cwfs} in Clairambault and Dybjer \cite{clairambault:london}). We will then study the
relationship with cartesian and cartesian closed categories. (By cartesian category we here mean categories with finite products, whereas some of the literature including Johnstone \cite{johnstone2002sketches} uses this term for categories with finite limits.)

\subsection{Plain scwfs}

A \emph{simply-typed cwf (scwf)} is a cwf where the presheaf of types $\Ty :
\Cop \to \Set$ is constant, \emph{i.e.} forms a set $\Ty$ not depending on the
context and invariant under substitution.
We can thus simplify the definition as follows: 

\begin{definition}
An scwf consists of the following:
\begin{itemize}
\item A category $\C$ with a terminal object $1$.


\item A set $\Ty$.


\item A family of presheaves $\Tm_A : \Cop \to \Set$ for $A \in \Ty$.
(We also write $\Tm(\Gamma, A)$ for $\Tm_A(\Gamma)$.)
\item A \emph{context comprehension} operation which to
$\Gamma \in \C_0$ and $A \in \Ty$ assigns a context $\Gamma \cext A$ and
two projections
\[
\p_{\Gamma, A} : \Gamma \cext A \to \Gamma 
\qquad
\qquad
\q_{\Gamma, A} \in \Tm(\Gamma \cext A, A)
\]
satisfying the following universal property: for all $\gamma : \Delta \to
\Gamma$, for all $a \in \Tm(\Delta, A)$, there is a unique $\tuple{\gamma, a} :
\Delta \to \Gamma \cext A$ such that
\[
\p_{\Gamma,A} \circ \tuple{\gamma, a} = \gamma
\qquad
\qquad
\q_{\Gamma,A}[\tuple{\gamma, a}] = a
\]

We say that $(\Gamma\cext A, \p_{\Gamma, A}, \q_{\Gamma, A})$ is a
\emph{context comprehension} of $\Gamma$ and $A$.
\end{itemize}
\end{definition}

\begin{remark}
Context comprehension for scwfs amounts to a representation $\Gamma \cext A$
of the presheaf
$$
\C(-,\Gamma) \times \Tm_A(-) : \C^\op \to \Set
$$
for all $\Gamma \in \Cobj$ and $A \in \Ty$. 
\end{remark}

An scwf is a particular kind of cwf, and a \emph{(strict) scwf-morphism} is simply a
(strict) cwf-morphism between scwfs. Small scwfs and strict cwf-morphisms form a category $\Scwf$.

$\Scwf$ has an initial object, but it is not very interesting since its
set of types is empty and its base category is restricted to the terminal object.
Therefore we fix a set
$\B$ of \emph{basic types} and consider the category $\B$-$\Scwf$ where objects
are small scwfs $(\C, \Ty_\C, \Tm_\C)$ together with an interpretation function
$\intr{-}_\C : \B \to \Ty$. Morphisms are strict scwfs-morphisms that
commute with the interpretation.

\begin{proposition}\label{prop:free_scwf}
For all sets $\B$ the category $\B$-$\Scwf$ has an initial object.
\end{proposition}
\begin{proof}
The initial $\B$-scwf, also called the \emph{free scwf over a set of types
$\B$}, can be constructed in much the same two ways as we constructed initial ucwfs. 

We can either proceed as in Construction 1 where the operations become introduction rules for inductively generating the objects, and terms, and the equations inductively generate an equivalence relation. 

Alternatively, we can proceed as in Construction 2, and construct a typed version of the category $\Ncat$ of Proposition 
\ref{prop:ucwf_initial}. For the second construction we let $\Ty = \B$. Then we define
\begin{eqnarray*}
\Cobj &=& \List(\Ty)\\
\C(\Gamma,[A_1,\ldots,A_n]) &=& \Tm(\Gamma,A_1) \times \cdots \times
\Tm(\Gamma,A_n)
\end{eqnarray*}
where $\Tm(\Gamma,A) = \{\var_n(i) \mid A = A_i \}$ containing the $i$th
\emph{variable}, where $\Gamma = [A_1, \dots, A_n]$. Moreover, we define the
\emph{projections} as
\begin{eqnarray*}
\p_{\Gamma, A} &=& (\var_{n+1}(1),\ldots,\var_{n+1}(n))\\
\q_{\Gamma, A} &=& \var_{n+1}(n+1)
\end{eqnarray*}
and substitution 
$$\var_n(i) [(a_1, \dots, a_{n})]= a_i$$
\end{proof}

Just as for ucwfs, the free scwf over
a set $\B$ is contextual. In analogy with Section \ref{subsec:contextual_ucwfs} we could
relate scwfs to coloured cartesian operads (multicategories) and multi-sorted Lawvere
theories, but we omit the unsurprising details. Instead, we discuss their relationship with \emph{cartesian
categories}.

\subsection{Finite products as structure}
\label{subsec:prod_struct}

Cartesian categories (categories with finite products) are categories with a
terminal object and binary cartesian products.
Straightforward as it seems, this definition hides
some subtleties which are put to the forefront when one considers the
associated notions of \emph{morphism} between cartesian categories.
Namely, does the mere \emph{existence} of a product for any two objects suffice
to obtain a cartesian category, or are the finite products part of the
\emph{data} of a cartesian category? Texts in category theory adopt one view or
the other, not always explicitely. In Lambek and Scott's book, the latter view
is explicitely adopted; and morphisms between cartesian categories must
preserve this explicit data \emph{on the nose}.

\begin{definition}\label{def:cc_structure}
A \emph{cartesian category (with structure)} consists of the following:
\begin{itemize}
\item A category $\C$ with a terminal object $1$, where $\tuple{}_X : X \to 1$ denotes the unique arrow into it.
\item A \emph{product operation} which to any $A, B \in \C_0$ assigns an object
$A \times B \in \C_0$ and two projections
\[
\fst_{A, B} : A \times B \to A 
\qquad
\qquad
\snd_{A, B} : A \times B \to B
\]
satisfying the following universal property: for all $a : X \to A$, for all $b
: X \to B$, there is a unique $\tuple{a, b} : X \to A \times B$ such that
\[
\fst_{A, B} \circ \tuple{a, b} = a 
\qquad
\qquad
\snd_{A, B} \circ \tuple{a, b} = b\,.
\]
\end{itemize}

A \emph{strict cartesian functor} from cartesian category $\C$ (leaving
implicit the other components) to $\D$ is a functor $F : \C
\to \D$ preserving the structure \emph{on the nose}, \emph{i.e.} for all $A, B
\in \C_0$, $F(A \times_\C B) = F(A) \times_\D F(B)$, and 
\[
F(\fst^\C_{A, B}) = \fst^\D_{F(A),F(B)}
\qquad
\qquad
F(\snd^\C_{A, B}) = \snd^\D_{F(A), F(B)}\,.
\]

We write $\CCs$ for the category of small cartesian categories (with structure)
and strict cartesian functors, preserving the structure on the nose.
\end{definition}

We will later consider another notion of cartesian category, where we are
content with the \emph{mere existence} of a terminal object and a cartesian product for any two
objects, and the only data part of the structure is the basic category
$\C$ -- there are no \emph{chosen products}. 
This, of course, constrains the maps: the notion of strict cartesian morphisms
as above would make no sense without chosen structure.

We observe in passing that from the preservation of projections and the
universal property, it follows directly that strict cartesian functors also
preserve tuples on the nose, in the sense that $F(\tuple{}_X) = \tuple{}_{F(X)}$ and $F(\tuple{a, b}) = \tuple{F(a),
F(b)}$.

\subsubsection{Finite product types}
We now wish to compare cartesian categories with structure, as above, with
scwfs. They are similar, but the main difference is that 
scwfs distinguish \emph{contexts} and \emph{types}; whereas cartesian
categories do not have this distinction.
In particular, when constructing an scwf from a cartesian category we must
recover the \emph{types} as the objects of $\C$. It follows that
the resulting scwfs support a finite
product operation on types. Looking for an equivalence, we define what
it means for an scwf to support finite product types and introduce a binary product type $A \times B$, and a unit
type $\None$.

\begin{definition}
An \emph{$\None$-structure} on an scwf $\C$ consists of a type $\None \in \Ty$,
and for each $\Gamma$ a term ${0_1} \in \Tm_\C(\Gamma, \None)$ such
that for all $c \in \Tm(\Gamma, \None)$, $0_1= c$.
\end{definition}


\begin{definition}\label{def:times_struct}
A \emph{$\times$-structure} on an scwf $\C$ consists of, for each 
$A, B \in \Ty$, a type $A\times B \in \Ty$ such that for all 
$\Gamma \in \C_0$ there are term formers
\begin{eqnarray*}
\fst_{\Gamma, A, B}(-) &:& \Tm(\Gamma, A \times B) \to \Tm(\Gamma, A)\\
\snd_{\Gamma, A, B}(-) &:& \Tm(\Gamma, A \times B) \to \Tm(\Gamma, B)\\
\tuple{-, -} &:& \Tm(\Gamma, A) \times \Tm(\Gamma, B) \to \Tm(\Gamma, A\times
B)
\end{eqnarray*}
such that
\begin{eqnarray*}
\fst(\langle a , b \rangle) &=& a \\
\snd(\langle a , b \rangle) &=& b \\
\langle \fst(c) , \snd(c) \rangle &=& c\\
\langle a , b \rangle [\gamma] &=& \langle a [\gamma], b  [\gamma] \rangle\,.
\end{eqnarray*}
\end{definition}
We thus have a type formation rule for $A\times B$ and term
formation rules for projections and pairs. The first three equations are the usual  rules for product types with surjective
pairing, while the last one states stability under substitution. 
By an \emph{scwf with finite product types} we mean an \emph{scwf with an
$\None$-structure and a $\times$-structure}. 

\begin{remark}
Having a $\times$-structure on an scwf $\C$ amounts to requiring that there is a
binary type former $\times$ and a natural isomorphism of presheaves
\[
\Tm_\C(-, A) \times \Tm_\C(-, B) \cong \Tm_\C(-, A\times B)\,.
\]
%
Likewise, an $\None$-structure corresponds to a type $\None$ and a natural
isomorphism between $\Tm_\C(-, \None)$ and the constant singleton presheaf.
\end{remark}

The product type structure on scwfs 
should be \emph{preserved} by morphisms.

\begin{definition}
If $\C, \D$ are scwfs with finite product types, then a strict scwf-morphism $F : \C \to \D$
\emph{(strictly) preserves $\None$-structure} if $F(\None^\C) = \None^\D$, and
$F(0_1^\C) = 0_1^\D$. Similarly, it \emph{(strictly) preserves
$\times$-structure} if $F(A\times^\C B) = F(A) \times^\D F(B)$, $F(\fst_{A,B}^\C(c))
= \fst_{FA, FB}^\D(F(c))$ and $F(\snd_{A, B}^\C(c)) = \snd_{FA, FB}^\D(F(c))$.

We write $\Scwf^{\None,\times}$ for the category where the objects are small scwfs with
product types, and morphisms are strict structure-preserving cwf-morphisms.
\end{definition}

Given an scwf with finite product types we construct a cartesian category as follows. 
First, we define the \emph{category of types and terms in context}.

\begin{definition}
Let $\C$ be an scwf with finite product types. We 
define the category $\Tycat(\Gamma)$ of \emph{types and terms in context $\Gamma \in \Cobj$} as
having: \emph{(1)} as objects, $\Ty$; \emph{(2)} as morphisms from $A \in \Ty$
to $B \in \Ty$, the terms $b \in \Tm(\Gamma \cext A, B)$. If 
$b \in \Tm(\Gamma \cext A, B)$ and $c
\in \Tm(\Gamma \cext B, C)$, then their composition is
\[
c \circ b = c[\tuple{\p_{\Gamma,A}, b}]
\]
with identity $\id_{\Gamma, A} = \q_{\Gamma, A} \in
\Tm(\Gamma \cext A, A)$.
\end{definition}
It follows that $\Tycat(\Gamma)$ is a
cartesian category with structure for all $\Gamma \in \C_0$.

\begin{lemma}
For any $\Gamma$, we let $\None$ be the chosen terminal object in $\Tycat(\Gamma)$.
For every $A, B \in \Ty$, we let their product be $A \times B  \in \Ty$ and the 
projections
\[
\fst(\q_{\Gamma, A\times B}) \in \Tycat(\Gamma)(A\times B, A) 
\qquad
\qquad
\snd(\q_{\Gamma, A\times B}) \in \Tycat(\Gamma)(A\times B, B)
\]
\end{lemma}

We omit the (straightforward) proof. 
This entails that given an scwf with finite product types $\C$, there is a canonical
cartesian category with structure, the cartesian category of closed
terms $\Tycat_\C(1_\C)$. 

\begin{proposition}
There is a functor
\[
\CF : \Scwf^{\None, \times} \to \CCs
\]
which to any scwf with products $\C$ associates $\Tycat_\C(1_\C)$, and to
any structure-preserving cwf-morphism $F : \C \to \D$ associates $\CF(F) :
\Tycat_\C(1_\C) \to \Tycat_\D(1_\D)$ given by the action of $F$ on types and terms.
\end{proposition}

\subsubsection{From cartesian categories to scwfs}

Since $\CF$ forgets the structure of
contexts and only remembers closed types, we expect a functor $\LF$ in the opposite
direction to somehow reconstruct contexts. There are two natural candidates for
this. The first is to reverse the effect of $\CF$ by reconstructing
the context formally, in an operation analogous to the construction of the
cartesian category of \emph{polynomials} in Lambek and Scott. We will
detail this below. In this way we do not directly get an equivalence, because
if $\C$ is an arbitrary scwf, the contexts of $\LF\CF(\C)$ are
\emph{generated inductively} from types. However, we get an equivalence for
\emph{contextual} scwfs.

The other option is to let the category of contexts of $\LF(\C)$ be $\C$, reflecting the dual role of objects in cartesian categories as both
contexts and types. Context comprehension is defined
via finite products. As simple as it looks, this construction does \emph{not}
yield an equivalence even when restricting scwfs. (It would if one restricted to
democratic scwfs such that for each $\Gamma \in \C_0$, we have
$\Gamma = 1\cext \overline{\Gamma}$, but this is not a natural hypothesis since it is
not satisfied by the syntax). It is, however, behind the
\emph{biequivalence} between scwfs and cartesian categories as property that we shall discuss in the following subsection.
\begin{definition}
If $\C$ is a cartesian category with structure, we define an scwf $\LF(\C)$ analogously to the
free scwf in Proposition \ref{prop:free_scwf}. The set of \emph{types} is $\Ty
= \C_0$. We define
\begin{eqnarray*}
\LF(\C)_0 &=& \List(\C_0)\\
\LF(\C)(\Gamma, [A_1, \dots, A_n]) &=& \C(\Pi \Gamma, A_1) \times \dots \times
\C(\Pi \Gamma, A_n)
\end{eqnarray*}
where $\Pi [B_1, \dots, B_m] = (\dots (B_1 \times B_2) \dots \times
B_m)$ and $\Pi [] = 1$. The \emph{terms} are 
\[
\Tm(\Gamma, A) = \C(\Pi \Gamma, A).
\]
Substitution is defined as 
\[
a[(\gamma_1, \dots, \gamma_n)] = a \circ \tuple{\tuple{\dots\tuple{\gamma_1,
\gamma_2}, \dots, },\gamma_n} \in \Tm(\Delta, A)
\]
for $(\gamma_1, \dots, \gamma_n) : \Delta \to [A_1, \dots, A_n]$ and $a \in
\Tm(\Gamma, A)$ and \emph{composition} in $\LF(\C)$ by
$(a_1, \dots, a_n) \circ \gamma = (a_1[\gamma],
\dots, a_n[\gamma])$. 

For $\Gamma = [A_1, \dots, A_{n}]$ and $1\leq i \leq n$, we
write $\var_n(i) \in \Tm(\Gamma, A_i)$ for the corresponding variable, obtained
as the $n$-ary projection.
For $\Gamma = [A_1, \dots, A_{n}] \in \LF(\C)_0$ and $A \in \Ty$ we let $\Gamma \cext A = [A_1, \dots, A_{n},A]$.
The \emph{projections} are defined by
\begin{eqnarray*}
\p_{\Gamma, A} &=& (\var_{n+1}(1), \dots, \var_{n+1}(n))\\
\q_{\Gamma, A} &=& \var_{n+1}(n+1)\,.
\end{eqnarray*}

We thus get an scwf with finite product types: the
$\None$-structure is the terminal object of $\C$ and
$A\times B$ is given by the cartesian product of $\C$. If $c \in
\Tm(\Gamma, A\times B)$, the projections
\[
\fst_{A, B}(c) = \fst_{A, B} \circ c
\qquad
\qquad
\snd_{A,B}(c) = \snd_{A, B} \circ c
\]
are immediate. (Note the overloading of $\fst$ and $\snd$.)
\end{definition}

We observe that our construction yields that $\CF(\LF(\C)) = \C$
for each cartesian category with structure $\C$.
This construction can be lifted to a functor
\[
\LF : \CCs \to \Scwf^{\None, \times}
\]
where, given $F : \C \to \D$, $\LF(F) : \LF(\C) \to \LF(\D)$ is obtained by letting
$F$ act component-wise. It is thus clear that all structure is
preserved. Finally, $\CF$ and $\LF$ do \emph{not} yet form an equivalence.
Indeed, $\LF(\CF(\C))$ is always \emph{contextual} since its contexts are
inductively generated, whereas $\C$ might not be. For instance, any context of
$\C$ which is not obtained as an iterated context extension of types is not
accounted for in $\LF(\CF(\C))$. However, we have:

\begin{theorem}\label{th:eq_cc}
The functors $\CF$ and $\LF$ form an equivalence of categories:
\[
\xymatrix@C=60pt{
\CCs	\ar@/^1pc/[r]^{\LF}&
\Scwf^{\None,\times}_{\ctx}
	\ar@/^1pc/[l]^{\CF}
}
\]
\end{theorem}
\begin{proof}
It only remains to observe that for any
contextual scwf $\C$ with finite products types we have the isomorphism $\C \cong
\LF(\CF(\C))$, where this isomorphism sends a context $1\cext A_1
\cext \dots \cext A_n$ to $[A_1, \dots, A_n]$.
\end{proof}


\subsection{Finite products as property}

\subsubsection{Cartesian categories as property} 

We now define a notion of cartesian category where finite products are defined as a property of a category.

\begin{definition}
Let $\C$ be a category. It is \emph{cartesian (as a property)} if
\emph{there exists} a terminal object in $\C$, and if for any two objects $A, B
\in \C_0$, \emph{there exists} a cartesian product of $A$ and $B$, \emph{i.e.}
a triple $(P, \pi, \pi')$ with $P\in \C_0$, $\pi : P \to A$, $\pi' : P
\to B$ satisfying the usual universal property. 
\end{definition}

By ``there exists'', we mean \emph{mere existence}. The choice
of terminal objects and the assignment of $(P, \pi, \pi')$ from $A$ and $B$ are not part of
the structure. A cartesian
category is just a particular kind of category with no additional data.
Likewise, cartesian functors may be defined as:

\begin{definition}
A functor $F : \C \to \D$ is
\emph{cartesian} if the image of a terminal object is terminal, and the image
of a product $(P, \pi, \pi')$ is a product $(FP, F\pi, F\pi')$.
We let $\CCp^2$ be the $2$-category of small cartesian categories (with 
property) as objects, cartesian functors as $1$-cells, and natural
transformations as $2$-cells.
\end{definition}
Thus cartesian functors are just certain functors with no additional data. We regard $\CCp^2$ as a $2$-category because we will prove a \emph{biequivalence} rather than an equivalence. Indeed, scwfs
have chosen structure while cartesian categories (with property) do not. Going
from an scwf to a cartesian category and back the structure is forgotten and
then \emph{chosen} again, and we cannot recover the
original scwf up to isomorphism, only up to equivalence.

\begin{remark} 
We could also introduce a notion of {\em scwf with property}, where context comprehension is only a property and not part of the structure. We will discuss this option briefly in the chapter about full cwfs.
\end{remark}

\subsubsection{Pseudo scwf-morphisms} 
Previously, we defined a notion of strict scwf-morphism, but this does not match the notion of 
cartesian functor as property. To address this mismatch we need a notion of \emph{pseudo cwf-morphism} where structure is only preserved \emph{up to
isomorphism}. 
\begin{definition}
A \emph{pseudo scwf-morphism} from the scwf $\C$ to the scwf $\D$ consists of a functor
$F : \C \to \D$, a function $F^\Ty : \Ty^\C \to \Ty^\D$, and a family
\[
F^\Tm_A : \Tm^\C_A(-) \tto \Tm^\D_{F^\Ty A}(-)
\]
of natural transformations. We write $F^\Tm_{\Gamma, A}$ for
the component of $F^\Tm_A$ on $\Gamma$.

These data are subject to the conditions that \emph{(1)} $F1$ is terminal in
$\D$, and \emph{(2)} for all $\Gamma \in \C_0$, $A \in \Ty^\C(\Gamma)$, the
triple $(F(\Gamma\cext A), F(\p_{\Gamma, A}), F_{\Gamma, A}^\Tm(\q_{\Gamma,
A}))$ is a context comprehension of $F\Gamma$ and $F^\Ty(A)$ in $\D$. 
\end{definition}

For strict cwf-morphisms, the triple $(F(\Gamma\cext A), 
F(\p_{\Gamma, A}), F_{\Gamma, A}^\Tm(\q_{\Gamma, A}))$ must coincide with the context comprehension of $F\Gamma$ and $F^\Ty(A)$
chosen by the scwf structure of $\D$. Here, we drop that assumption.

For related reasons, the equivalence will work slightly differently than in
Section \ref{subsec:prod_struct}. There the
equivalence followed the slogan ``the cartesian category corresponding to a
scwf is its category of closed terms''. However, in Theorem~\ref{th:eq_cc}
products in a cartesian category $\C$ are used in a central way in the
definition of terms in $\LF(\C)$. These can be chosen once and for all, but
the preservation of products up to isomorphism leads to an unwieldy
definition of the functorial action of $\LF$. Instead, we adopt
a simpler approach under the slogan ``the cartesian category
corresponding to an scwf is its base category''. With the aim of getting a biequivalence rather than an equivalence, the requirement of
\emph{contextuality} will be replaced by the weaker \emph{democracy}. This additional structure must then also be preserved.

\begin{definition}
A \emph{democratic} pseudo scwf-morphism between democratic scwfs $\C$ and $\D$ additionally has, for each $\Gamma \in \C_0$, as
isomorphism $\d_\Gamma : F^\Ty(\overline{\Gamma}) \cong \overline{F \Gamma}$ in
the category $\Tycat^\C(1)$, subject to a coherence diagram expressing that
$F(\gamma_\Gamma) = \gamma_{F\Gamma}$ modulo some transports\footnote{Using
the existing structure one can define a canonical morphism from $\overline{F
\Gamma}$ to $F^\Ty(\overline{\Gamma})$ in $\Tycat^\C(1)$. The coherence diagram
amounts to $d_\Gamma$ being its inverse -- hence $F$ being democratic amounts
to requiring that the canonical map from $\overline{F \Gamma}$ to
$F^\Ty(\overline{\Gamma})$ is an iso.} (see \cite{ClairambaultD14}, p.19).  
\end{definition}

A democratic scwf always has finite product types defined by $\None = \overline{1}$ and $A\times B = \overline{1\cext A \cext B}$. So our
biequivalence holds without adding these as extra structure.

\begin{definition}
We write $\Scwf_{\dem}^2$ for the $2$-category with democratic small scwfs as
objects, democratic pseudo scwf-morphisms as $1$-cells, and natural
transformations between the underlying functors as $2$-cells.
\end{definition}

One may wonder why the definition of $2$-cells does not have a type component. Such a
component could be added, but then the
biequivalence requires a coherence law which makes it
redundant (see \cite{castellan:lmcs}, p.7 and the expanded
discussion in Appendix B).

\subsubsection{The biequivalence} 
\label{subsubsec:biequivalence_st}
We now provide the components of the
biequivalence. First we observe that there is a forgetful $2$-functor
\[
\CF^2 : \Scwf_\dem^2 \to \CCp^2
\]
which maps a democratic scwf to its base category, a democratic pseudo
cwf-morphism to its base functor, and leaves $2$-cells unchanged. That this
is well-defined relies on the facts that \emph{(1)} if $\C$ is a democratic scwf,
then, the base category $\C$ is cartesian (with property), because products in $\C$
may be defined as $\Gamma \times \Delta = \Gamma \cext \overline{\Delta}$; \emph{(2)}
if $F$ is a democratic pseudo cwf-morphism, then the fact that the base functor
preserves products follows from preservation of context comprehension and democracy.

We now construct a $2$-functor in the other direction.

\begin{proposition}
There is a $2$-functor $\LF^2 : \CCp^2 \to \Scwf_\dem^2$.
\end{proposition}
\begin{proof}
In every cartesian category (with property) $\C$ we choose a terminal object
$1$ and a product $(A\times B, \fst_{A,B}, \snd_{A,B})$ for every $A, B \in
\C_0$. 

We now turn $\C$ into an scwf in the following way. 
Its category of contexts is $\C$. Its set of types is $\Ty = \C_0$. If $\Gamma
\in \C_0$ and $A \in \Ty$, the terms are $\Tm_\C(\Gamma, A) = \C(\Gamma, A)$.
Context comprehension is given by finite products. Democracy is the
isomorphism $\Gamma \cong 1\times \Gamma$. Likewise, a cartesian functor $F :
\C \to \D$ is extended to types and terms in the obvious way, and yields a
democratic pseudo cwf-morphism. 
\end{proof}


\begin{theorem}\label{th:biequivalence_cc}
The $2$-functors $\CF^2$ and $\LF^2$ form a biequivalence of $2$-categories:
\[
\xymatrix@C=60pt{
\CCp^2    \ar@/^1pc/[r]^{\LF^2}&
\Scwf^2_\dem
        \ar@/^1pc/[l]^{\CF^2}
}
\]
\end{theorem}
\begin{proof}
We have $\CF^2\LF^2 = \Id_{\CCp}$. To obtain a biequivalence we must construct
pseudonatural transformations of $2$-functors (or pseudofunctors)
\[
\xymatrix{
\mathrm{\Id_{\Scwf^2_\dem}}
\ar@<0.5ex>[r]^{\eta}&
\LF^2\CF^2
\ar@<0.5ex>[l]^{\epsilon}
}
\]
which are inverses up to invertible modifications. Concretely, we construct,
for each democratic scwf $\C$, pseudo cwf-morphisms 
$\eta_\C : \C \to \LF^2\CF^2 \C$ and $\epsilon_\C : \LF^2\CF^2 \C \to \C$,
both with the identity functor as base component.
Besides we have $\eta_\C^\Ty(A) = 1.A$ and $\epsilon_\C^\Ty(\Gamma) =
\overline{\Gamma}$ on types, and 
\[
{\eta_\C^\Tm}_{\Gamma, A}(a) = \tuple{\tuple{}, a}
\qquad
{\epsilon_\C^\Tm}_{\Delta, \Gamma}(\gamma) = \q_{1, \overline{\Gamma}}[\gamma_\Gamma \circ \Gamma]
\]
on terms. These two pseudo cwf-morphisms are pseudonatural in $\C$, and
form an equivalence in $\Scwf^2_\dem$: $\eta_\C \circ \epsilon_\C$ and
$\mathrm{id}_{\CF^2\LF^2\C}$ (resp. $\epsilon_\C \circ \eta_\C$ and
$\mathrm{id}_{\C}$) are related by invertible $2$-cells with the identity as
component. Finally, by unfolding definitions it follows that the
invertible $2$-cells satisfy the coherence condition for modifications between
pseudonatural transformations.
\end{proof}

Together, Theorems \ref{th:eq_cc} and \ref{th:biequivalence_cc} give two
different points of view on the correspondence between scwfs and cartesian
categories. It is interesting that we need to use a biequivalence 
not just for Martin-L\"of's type theory and locally cartesian closed categories
as in \cite{ClairambaultD14}, but already for cartesian categories if
we omit chosen structure. For cartesian categories we can build both an
equivalence (with structure) and a biequivalence (with property), whereas it seems
that only the latter is possible for finitely complete categories and locally
cartesian closed categories. 
\subsection{Adding function types}

Before going on to the dependently typed case, we mention how to add
function types to the previous equivalences and biequivalences. First, we
recall:

\begin{definition}
A \emph{cartesian closed category (with structure)} is
a cartesian category (with structure) and an operation which to
$A, B \in \C_0$ assigns an object $A\tto B \in \C_0$ with an \emph{evaluation}
\[
\ev_{A, B} : (A\tto B) \times A \to B
\]
such that for all $f : C \times A \to B$, $\exists! h :
C\to A\tto B$ s.t. $\ev_{A, B} \circ (h \times A) = f$.

We write $\CCCs$ for the category having small cccs (with structure) as objects, and
as morphisms, functors $F : \C \to \D$ preserving the structure on the nose,
\emph{i.e.} $F(A\tto^\C B) = FA \tto^\D FB$, and $F(\ev^\C_{A, B}) =
\ev^\D_{FA,FB}$.
\end{definition}

Likewise, we may add function types to scwfs in the following way.

\begin{definition}\label{def:tto_structure}
A \emph{$\tto$-structure} on an scwf $\C$ consists of, for each $\Gamma \in
\C_0$ and $A, B\in \Ty$, a type $A\tto B$ together with term formers
\begin{eqnarray*}
\lambda_{\Gamma,A,B} &:& \Tm(\Gamma\cext A, B) \to \Tm(\Gamma, A \arrow B)\\
\ap_{\Gamma,A,B} &:& \Tm(\Gamma, A \arrow B) \times\Tm(\Gamma,A) \to
\Tm(\Gamma,B)
\end{eqnarray*}
s.t., for $a \in \Tm(\Gamma,A),b \in \Tm(\Gamma\cext A,B),c \in \Tm(\Gamma, A
\arrow B)$,$\gamma \in \C(\Delta,\Gamma)$:
\begin{eqnarray*}
\lambda_{\Gamma,A,B}(b)[\gamma] &=& \lambda_{\Delta,A,B}(b[\tuple{\gamma \circ
\p_{\Delta,A},\q_{\Delta,A}}]) \\
\ap_{\Gamma,A,B}(c,a)[\gamma]&=& \ap_{\Delta,A,B}(c[\gamma],a[\gamma])\\
\ap_{\Gamma,A,B}(\lambda_{\Gamma,A,B}(b),a) &=& b[\tuple{\id_\Gamma,a}]\\
\lambda_{\Gamma,A,B}(\ap_{\Gamma \cext A,A,B}(c[\p_{\Gamma,A}],\q_{\Gamma,A})) &=&
c\,.
\end{eqnarray*}

If $\C, \D$ have $\tto$-structure, a (strict) cwf-morphism $F : \C \to \D$
\emph{preserves} it if $F(A\tto^\C B) = FA \tto^\D FB$, and $F(\ap_{\Gamma, A,
B}^\C(c, a)) = \ap_{F\Gamma, FA, FB}^\D(F(c), F(a))$.
\end{definition}

\begin{remark}
An $\arrow$-type structure on $(\C,\Ty,\Tm)$ is equivalent to the data of
a binary type former $\arrow$ and a natural isomorphism of preheaves 
\[
\Tm(-\cext A,B) \cong \Tm(-, A\arrow B)\,,
\]
\emph{i.e.} a representation of $\Tm(-\cext A,B)$. 

For $\gamma : \Delta \to \Gamma$, the functorial action $\Tm(\gamma
\cext A, B)$ takes $b \in \Tm(\Gamma\cext A, B)$ to $b[\tuple{\gamma\circ \p,
\q}] \in \Tm(\Delta \cext A, B)$.
We recover
$
\ap_{\Gamma, A, B}(f,a) = \lambda^{-1}(f)[\tuple{\id,a}]
$
and derive the $\beta$ and $\eta$-rules, and vice versa.
\end{remark}

Let us write $\Scwf_\ctx^{\None, \times, \tto}$ for the category having as
objects small contextual scwfs with an $\None$-structure, a $\times$-structure and a
$\tto$-structure, and as morphisms the strict cwf-morphisms preserving this
structure on the nose. 

\begin{theorem}\label{th:eq_strict_cc}
The categories $\CCCs$ and $\Scwf_\ctx^{\None, \times, \tto}$ are equivalent.
\end{theorem}
\begin{proof}
Straightforward extension of Theorem \ref{th:eq_cc}, which boils down to the
definition of evaluation from application and vice versa.
\end{proof}

This is our version of one of the main
results of the Lambek and Scott book, namely Theorem 3.11, the equivalence
between simply-typed $\lambda$-calculi and cartesian closed categories, where
$\Scwf_\ctx^{\None, \times, \tto}$ plays the role of the category of
typed $\lambda$-calculi. And indeed, the simply-typed $\lambda$-calculus arises as
the free scwf with $\tto$-structure over a set of types. If $\B$ is a set of
basic types, we consider $\B$-$\Scwf_\ctx^{\tto}$ to be the category with
objects scwfs with $\tto$-structure together with an interpretation function
$\intr{-}_\C : \B \to \Ty$; and as morphisms the strict scwf-morphisms that
preserve $\tto$-structure and commute with the interpretation. Then we have:

\begin{proposition}
For all sets $\B$, the category $\B$-$\Scwf^{\tto}$ has an initial object.
\end{proposition}
\emph{Construction 1.}
Just as for ucwfs and scwfs, we can immediately turn the definition of scwf
with $\tto$-structure into an inductive definition of a free such. The resulting theory is a well-scoped variable-free version of the typed
$\lambda\sigma$-calculus with base types in $\B$.

\emph{Construction 2.} Alternatively, we can construct an scwf with $\tto$-structure
free over $\B$ from a well-scoped version of the simply-typed
$\lambda\beta\eta$-calculus. The construction follows that of Proposition
\ref{prop:free_scwf}, except that the terms are now inductively defined with the
following three constructors: 
\begin{eqnarray*}
\var_n(i)&:& \Tm([A_1, \ldots, A_i, \ldots, A_n],A_i)\ \ \ (i \in \Fin(n))\\
\lambda_{\Gamma,A,B} &:& \Tm(\Gamma \cext A, B) \to \Tm(\Gamma, A \arrow B)\\
\ap_{\Gamma,A,B} &:& \Tm(\Gamma, A \arrow B) \times\Tm(\Gamma,A) \to
\Tm(\Gamma,B)\,.
\end{eqnarray*}

The definition of substitution is then extended with 
\begin{eqnarray*}
\lambda_{\Gamma,A,B}(b)[\gamma] &=& \lambda_{\Delta,A,B}(b[\tuple{\gamma \circ
\p_{\Delta,A},\q_{\Delta,A}}]) \\
\ap_{\Gamma,A,B}(c,a)[\gamma]&=& \ap_{\Delta,A,B}(c[\gamma],a[\gamma])
\end{eqnarray*}
for $\gamma \in \C(\Delta,\Gamma)$.
Finally we quotient with the equivalence relation generated by $\beta\eta$-conversion:
\begin{eqnarray*}
\ap_{\Gamma,A,B}(\lambda_{\Gamma,A,B}(b),a) &\sim& b[\tuple{\id_\Gamma,a}]\\
\lambda_{\Gamma,A,B}(\ap_{\Gamma \cext A,A,B}(c[\p_{\Gamma,A}],\q_{\Gamma,A})) &\sim&
c
\end{eqnarray*}

The two constructions and the proof of their equivalence have been formalised in Agda by Brilakis \cite{Brilakis18}.

We can of course construct free objects in $\Scwf^{\None,\times,\arrow}$ and
other categories of scwfs with extra type structure in a similar way. Likewise,
we can show a \emph{biequivalence} between the $2$-category of cccs \emph{as
property} rather than structure, and the extension of the $2$-category
$\Scwf^2_\dem$ where scwfs additionally have a $\tto$-structure, which is
preserved up to isomorphism. We omit the details since they are similar to those in the proof
of Theorem \ref{th:biequivalence_cc}.

\section{Dependently typed categories with families}

We now return to full dependently typed cwfs. 
After discussing alternative definitions, we show an explicit
construction of a free cwf. We then add the type formers $\Id,
\Sigma,$ and $\Pi$ and give an overview of the biequivalence theorems
in Clairambault and Dybjer \cite{ClairambaultD11,ClairambaultD14}.
Finally we outline the construction of a bifree lccc and the proof of
undecidability of equality \cite{castellan:tlca2015,castellan:lmcs} in this lccc.

\subsection{Plain cwfs}
\label{subsec:plaincwfs}


%

\subsubsection{Alternative definitions}

%

\noindent\textbf{Context comprehension via representable presheaves.}
First, observe that the family valued presheaf $T : \Cop \to \Fam$ may
equivalently be given by two $\Set$-valued presheaves
\[
\Ty : \Cop \to \Set
\qquad
\qquad
\Tm : (\int^\C \Ty)^\op \to \Set
\]
where $\int^\C \Ty$ is the category of elements of
the first presheaf.  Context comprehension for
full cwfs is equivalent to requiring that for all $\Gamma \in \Cobj$ and $A \in \Ty(\Gamma)$, there is a representation $\Gamma \cext A$ of the presheaf
\begin{eqnarray*}
\Delta &\mapsto& \sum_{\gamma \in \C(\Delta,\Gamma)}
\Tm(\Delta,A[\gamma])\\
\delta &\mapsto& (\gamma,a) \mapsto (\gamma \circ \delta, a [ \delta ])
\end{eqnarray*}

\begin{remark} 
Alternatively, we could consider a weaker notion of cwf with {\em context comprehension as a property} rather than structure. We then only require that the above presheaves are {\em representable}, that is, we only require the mere existence of the representing presheaves.
\end{remark}
%

\bigskip

\noindent\textbf{Natural models.}
More radically, Awodey \cite{Awodey18} and Fiore \cite{Fiore12}
propose to replace the $\Fam$-valued presheaf $T : \C^\op \to \Fam$
by two set valued presheaves 
\[
\Ty : \Cop \to \Set
\qquad
\qquad
\Tm : \Cop \to \Set
\]
and a natural transformation $\mathsf{typeof} : \Tm \natto \Ty$.
One can then define context comprehension in terms of
representable natural transformations. Let $\Y : \C \to
\Set^{\Cop}$ be the Yoneda embedding. A natural transformation $\sigma : G
\natto F$ between presheaves on $\C$ is representable in the sense of
Grothendieck, if for all $C \in \Cobj$ and $c \in F(C)$, there are $D
\in \Cobj, p \in \C(D,C)$, and $d \in G(D)$, such that the following
diagram in the category of presheaves is a pullback: 
$$
\xymatrix{\Y(D) \ar@{=>}[r]^{d}
\pb{315}
\ar@{=>}[d]_{\Y(p)}
& G \ar@{=>}[d]^{\mathrm{\sigma}}\\
\Y(C) \ar@{=>}[r]_{c} & F
}
$$
where $c$ and $d$ in the diagram are shorthand for the corresponding respective natural transformations $F(-)(c)$ and $G(-)(d)$ from the Yoneda lemma.

Hence, $\mathsf{typeof} : \Tm \natto \Ty$ is representable provided
for all $\Gamma \in \Cobj$ and $A \in \Ty(\Gamma)$,  
there is $\Gamma\cext A \in \Cobj$, $\p_{\Gamma,A} \in
\C(\Gamma\cext A,\Gamma)$, and $\q_{\Gamma,A} \in\Tm(\Gamma\cext A)$, such that
the following diagram is a pullback:
$$
\xymatrix@C=30pt{
\C(-,\Gamma\cext A)  \ar@{=>}[r]^(.6){\q_{\Gamma,A}[-]} \ar@{=>}[d]_{\p_{\Gamma,A} \circ -}
		\pb{315} & \Tm \ar@{=>}[d]^{\mathsf{typeof}}\\
\C(-,\Gamma) \ar@{=>}[r]_{A[-]} & \Ty
}
$$
We emphasize that the function that maps $\Gamma$ and $A$ to the triple $\Gamma \cext A, \p_{\Gamma,A}, \q_{\Gamma,A}$ is part of the structure -- the natural transformation $\mathsf{typeof} : \Tm \natto \Ty$ is \emph{represented}.

Awodey \cite{Awodey18} and Newstead \cite{Newstead:phd} study this notion under the name
\emph{natural models}. We refer to their papers for further development of
the theory.
Note that this approach suggests an essentially algebraic view of
cwfs rather than a generalized algebraic one.
Here we have only one sort $\Tm(\Gamma)$ containing \emph{all} terms
$a$ of some type $A = \mathsf{typeof}(a) \in \Ty(\Gamma)$. In other
words terms are ``fibred'' over types. 

\begin{remark}
Alternatively, we get a notion of natural model with context comprehension as property if we only require the mere representability of the natural transformation $\mathsf{typeof}$. This notion is considered by Ahrens, Lumsdaine, and Voevodsky \cite{AhrensLV18}.
\end{remark}
\bigskip

\noindent\textbf{Categories with attributes.} There is a certain
redundancy in the definition of categories with families, since we can show (using context comprehension) that \emph{terms} are in one-to-one
correspondence with certain morphisms of the base category:
$$
\Tm(\Gamma,A) \cong \{ \gamma \in \C(\Gamma,\Gamma \cext A)\ |\
\p_{\Gamma,A} \circ \gamma = \id_\Gamma\}
$$
In other words, terms in $\Tm(\Gamma, A)$ correspond to \emph{sections}
of $\p_{\Gamma, A}$, the \emph{display map} for the type
$A$. 

We can thus remove the term part of cwfs and get the closely related notion of a
\emph{category with attributes} (cwa) \cite{HofmannM:intttl}. This consists of a category $\C$ with a
terminal object, a presheaf $\Ty : \Cop \to \Set$ for types and
substitutions, and an operation which 
given $\Gamma \in \C_0$ and $A\in \Ty(\Gamma)$ associates a
context $\Gamma \cext A$ and a ``display map'' $\p_{\Gamma, A} :
\Gamma \cext A \to \Gamma$; and for each $\gamma : \Delta \to
\Gamma$, a chosen \emph{pullback square}:
\[
\xymatrix{
\Delta \cext A[\gamma]  \ar[r]
	\pb{315}
\ar[d]_{\p_{\Delta\cext A[\gamma]}} & \Gamma \cext A  \ar[d]^{\p_{\Gamma,A}}\\
\Delta \ar[r]_{\gamma} & \Gamma
}
\]

This follows the idea of ``substitutions as pullbacks'' familiar in
categorical logic. This choice of pullbacks is finally
required to be \emph{split}, in the sense that the association of
substitutions to pullback squares is functorial. 

It is fairly easy to prove that categories with attributes are
equivalent to categories with families \cite{hofmann:cambridge}.
This proof and several other proofs relating models of dependent type
theory are formalized in the UniMath system by Ahrens, Lumsdaine,
and Voevodsky \cite{AhrensLV18}.

Categories with attributes predate categories with families. In fact
categories with families were originally introduced \cite{dybjer:torino} as a modification of cwas. The main point
of the change was to obtain a definition that can be expressed as a generalised
algebraic theory with a transparent connection with Martin-L\"of's explicit
substitution calculus formulation of dependent type theory.
This was achieved by making the family of
terms into an explicit part of the definition and formalize the sets
of types and terms and their substitution operations as a family
valued presheaf. In this way the pullback property of type
substitution could be removed from the definition since it can be
derived from the other part of the structure.

\subsubsection{A free cwf}

As in the previous sections, we next show how to
build a free cwf. Recall that in the simply-typed case, we built a free scwf over a
set $\B$ of basic types, and we could do the same here. However, for simplicity, and
because it suffices to prove the undecidability theorem, we will assume that there is
only one basic type $o$. The rules used for the construction of this plain free cwf (except
the rule for the base type) correspond to the general rules for dependent type
theory. This construction can then be extended to the case where we add rules for
specific type formers. 

The construction of initial ucwfs and free scwfs are rather immediate from the
definitions: simply turn the definitions into simultaneous inductive
definitions of the families of terms and substitutions, and then define
equality of terms and substitutions by another simultaneous inductive
definition.  Unfortunately, the construction of free full cwfs is no longer
as direct. What complicates the matter is the type-equality rule,
which means that typability of terms may depend on proofs of equality
of types. Thus we have to define equality of contexts, context
morphisms, types, and terms simultaneously with their elements. Apart
from this the recipe is similar: take the definition of the
generalised algebraic theory of cwfs and turn it into a mutual
inductive definition where all equality reasoning is made explicit.

To build the free cwf we first define raw contexts, raw substitutions, raw
types, and raw terms.  
\begin{eqnarray*}
\Gamma \in \RawCtx &::=& 1  \ |\ \Gamma\cext A\\
\gamma \in \RawSub \ &::=& \gamma \circ \gamma \ |\ \id_\Gamma \ |\ \langle\rangle_\Gamma \ |\ \p_{A} \ |\ \langle \gamma, a \rangle_A\\
A \in \RawTy &::=& o \ |\  A[\gamma]\\
a \in \RawTm &::=& a[\gamma] \ |\ \qI_A
\end{eqnarray*}
We then need to define the well-formed contexts and types, and the well-typed
substitutions and terms. In Martin-L\"of's substitution calculus these are defined by a system of inference rules for all the eight forms of judgments. Here we choose a more economical way, by only defining well-formed \emph{equal} contexts and types, and the well-typed \emph{equal} substitutions and terms. Thus we define four families of partial equivalence relations (pers), corresponding to the four forms of equality judgments, by a mutual inductive definition:
$$
\Gamma = \Gamma' \vdash
\qquad
\Gamma \vdash A = A' 
\qquad
\Delta \vdash \gamma = \gamma' : \Gamma
\qquad
\Gamma \vdash a = a' : A
$$
where $\Gamma, \Gamma' \in \RawCtx, \gamma, \gamma' \in \RawSub, A, A' \in \RawTy,$ and $a,a' \in \RawTm$. The basic judgment forms can then be defined as the reflexive instances of the pers:
\begin{itemize}
\item 
$\Gamma \vdash$ abbreviates $\Gamma = \Gamma \vdash$, 
\item
$\Gamma \vdash A$ abbreviates $\Gamma \vdash A = A$,
\item
$\Delta \vdash \gamma : \Gamma$ abbreviates  $\Delta \vdash \gamma = \gamma : \Gamma$, 
\item
$\Gamma \vdash a : A$ abbreviates $\Gamma \vdash a = a : A$.
\end{itemize}

The four families of partial equivalence relations (pers) are given by a simultaneous inductive definition with the following introduction rules:

\boxit[Per-rules]{
\begin{mathpar}
\scalebox{.92}{\inferrule
	{\Gamma = \Gamma' \vdash \\
         \Gamma' = \Gamma'' \vdash} 
	{\Gamma = \Gamma'' \vdash}		
}\and
\scalebox{.92}{\inferrule
	{\Gamma = \Gamma' \vdash} 
	{\Gamma' = \Gamma \vdash}
}\and
\scalebox{.92}{\inferrule
	{\Delta \vdash \gamma = \gamma' : \Gamma \\ 
	 \Delta \vdash \gamma' = \gamma'' : \Gamma} 
	{\Delta \vdash \gamma = \gamma'' : \Gamma}
}\and
\scalebox{.92}{\inferrule
	{\Delta \vdash \gamma = \gamma' : \Gamma} 
	{\Delta \vdash \gamma' = \gamma : \Gamma}
}\and
\scalebox{.92}{\inferrule
	{\Gamma \vdash A = A' \\
	 \Gamma \vdash A' = A''} 
	{\Gamma \vdash A = A''}
}\and
\scalebox{.92}{\inferrule
	{\Gamma \vdash A = A'} 
	{\Gamma \vdash A' = A}
}\and
\scalebox{.92}{\inferrule
	{\Gamma \vdash a = a' : A \\
	 \Gamma \vdash a' = a'' : A} 
	{\Gamma \vdash a = a'' : A}
}\and
\scalebox{.92}{\inferrule
	{\Gamma \vdash a = a' : A}
	{\Gamma \vdash a' = a : A}
}
\end{mathpar}
}

\boxit[Preservation rules for judgments]{
\begin{mathpar}
\scalebox{.92}{\inferrule
	{{\Gamma} = {\Gamma}' \vdashS \\
	 {\Delta} = {\Delta}' \vdashS \\
	 {\Gamma} \vdashS \gamma = \gamma': {\Delta}} 
	{{\Gamma}' \vdashS \gamma = \gamma' : {\Delta}'}
}\and
\scalebox{.92}{\inferrule
	{{\Gamma} = {\Gamma}' \vdashS \\
	 {\Gamma} \vdashS A = A'}
	{{\Gamma}' \vdashS A = A'}
}\and
\scalebox{.92}{\inferrule
	{{\Gamma} = {\Gamma}' \vdashS \\
	 \Gamma \vdash A = A' \\
         {\Gamma} \vdashS a = a' : A}
	{{\Gamma}' \vdashS a = a' : A'}
}
\end{mathpar}
}
\boxit[Congruence rules for operators]{
\begin{mathpar}
\scalebox{.92}{\inferrule
	{ } 
	{1 = 1 \vdash }
}\and
\scalebox{.92}{\inferrule
	{\Gamma = \Gamma' \vdash \\
	 \Gamma \vdash A = A'} 
	{\Gamma \cext A = \Gamma' \cext A'\vdash} 
}\and
\scalebox{.92}{\inferrule
	{\Gamma \vdash A=A' \\
	 \Delta \vdash \gamma = \gamma' : \Gamma} 
	{\Delta \vdash A[\gamma] = A'[\gamma']} 
}\and
\scalebox{.92}{\inferrule
	{\Gamma = \Gamma' \vdashS } 
	{\Gamma \vdashS \id_\Gamma = \id_{\Gamma'} : \Gamma} 
}\and
\scalebox{.92}{\inferrule
	{\Gamma = \Gamma' \vdashS}
	{\Gamma \vdash \emptySub_\Gamma = \emptySub_ {\Gamma'} : \emptyContext}
}\and
\scalebox{.92}{\inferrule
	{{\Gamma} \vdashS \delta = \delta' : {\Delta} \\
	 {\Delta} \vdashS \gamma = \gamma': {\Theta} } 
	{{\Gamma} \vdashS \gamma \circ \delta = \gamma' \circ \delta' : {\Theta}} 
}\and
\scalebox{.92}{\inferrule
	{\Gamma \vdash A = A'}
	{\Gamma \cext A \vdash \p_{A} = \p_{A'} : \Gamma}
}\and
\scalebox{.92}{\inferrule
	{\Gamma \vdash A = A' \\
	 \Delta \vdash \gamma = \gamma' : \Gamma \\
	 \Delta \vdash a = a' : A[\gamma]} 
	{\Delta\vdash \langle \gamma,a \rangle_A = \langle \gamma',a' \rangle_{A'} : \Gamma \cext A} 
}\and
\scalebox{.92}{\inferrule
	{\Gamma \vdash a = a' : A \\
	 \Delta \vdash \gamma = \gamma' : \Gamma} 
	{\Delta \vdash a[\gamma] = a'[\gamma'] : A[\gamma]} 
}\and
\scalebox{.92}{\inferrule
	{\Gamma \vdash A = A'}
	{\Gamma \cext A \vdash \qI_{A} = \qI_{A'} : A[\p_A]} 
}
\end{mathpar}
}
\boxit[Conversion rules]{
\begin{mathpar}
\scalebox{.92}{\inferrule
	{\Delta \vdash \theta : \Theta \\
	 \Gamma \vdash \delta : \Delta \\
	 \Xi \vdash \gamma : \Gamma} 
	{\Xi \vdash (\theta \circ \delta) \circ \gamma = \theta \circ (\delta \circ \gamma) : \Theta} 
}\and
\scalebox{.92}{\inferrule
	{\Gamma \vdashS \gamma : \Delta} 
	{\Gamma \vdash \gamma = \id_\Delta \circ \gamma : \Delta}
}\and
\scalebox{.92}{\inferrule
	{\Gamma \vdashS \gamma : \Delta} 
	{\Gamma \vdash \gamma = \gamma \circ \id_\Gamma : \Delta}
}\and
\scalebox{.92}{\inferrule
	{\Gamma \vdash A \\
	 \Delta \vdash \gamma : \Gamma \\
	 \Theta \vdash \delta : \Delta} 
	{\Theta \vdash A[\gamma \circ \delta] = (A[\gamma])[\delta]}
}\and
\scalebox{.92}{\inferrule
	{\Gamma \vdash A}
	{\Gamma \vdash A[\id_\Gamma] = A}
}\and
\scalebox{.92}{\inferrule
	{\Gamma \vdash a : A \\
	 \Delta \vdash \gamma : \Gamma \\
	 \Theta \vdash \delta : \Delta} 
	{\Theta \vdash a[\gamma \circ \delta] = (a[\gamma])[\delta] : (A[\gamma])[\delta]}
}\and
\scalebox{.92}{\inferrule
	{\Gamma \vdash a : A}
	{\Gamma \vdash a[\id_\Gamma] = a : A} 
}\and
\scalebox{.92}{\inferrule
	{\Gamma \vdash \gamma : 1} 
	{\Gamma \vdash \gamma = \emptySub_{\Gamma} : 1} 
}\and
\scalebox{.92}{\inferrule
	{\Gamma \vdash A \\
	 \Delta \vdash \gamma : \Gamma \\
	 \Delta \vdash a : A[\gamma]}
	{\Delta \vdash \p_A \circ \langle \gamma,a \rangle_A = \gamma : \Gamma} 
}\and
\scalebox{.92}{\inferrule
	{\Gamma \vdash A \\
	 \Delta \vdash \gamma : \Gamma \\
	 \Delta \vdash a : A[\gamma]} 
	{\Delta \vdash \qI_A[\langle \gamma,a \rangle_A] = a : A[\gamma]} 
}\and
\scalebox{.92}{\inferrule
	{\Delta \vdash \gamma : \Gamma \cext A}
	{\Delta \vdash \gamma = \langle \p_A \circ \gamma , \qI_A[\gamma] \rangle_A : \Gamma \cext A}
}
\end{mathpar}
}
These rules correspond to the \emph{general rules} for intuitionistic type theory, that is, the rules which are given before the rules for the type formers, see the discussion in Section \ref{subsubsec:inferencerules}.

Finally, we have a rule for the base type:
\boxit[Base type]{
\begin{mathpar}
\scalebox{.92}{\inferrule
	{ } 
	{1 \vdash o = o} 
}
\end{mathpar}
}

\begin{theorem}\label{th:free_cwf}
The cwf $\T$, defined with the following data:
\begin{itemize}
\item $\T_0 = \{ {\Gamma}\ |\ \Gamma \vdashS \} /\!\! =^c$, where
  ${\Gamma} =^c {\Gamma}'$ if ${\Gamma} = {\Gamma}' \vdashS$ is
  derivable. 
\item
  $\T([{\Gamma}],[{\Delta}]) = \{ \gamma\ |\ \Gamma \vdashS \gamma
  : {\Delta} \} /\!\! =^{\Gamma}_{\Delta}$
  where $\gamma =^{\Gamma}_{\Delta} \gamma'$ iff
  ${\Gamma} \vdashS \gamma = \gamma' : {\Delta}$ is derivable. Note that
  this makes sense since it only depends on the equivalence classes $[ \Gamma ]$ and $[ \Delta ]$ of
  $\Gamma$ and $\Delta$ (morphisms and morphism equality are preserved by object
  equality).
\item $\Ty_\T([{\Gamma}]) = \{ A\ |\ \Gamma \vdashS A
  \}/=^{\Gamma}$ where $A =^{\Gamma} B$ if $\Gamma \vdashS A =
  B$.
\item $\Tm_\T([{\Gamma}],[A]) = \{ a\ |\ \Gamma \vdashS a: A\} / =^{\Gamma}_A$ where $a =^{\Gamma}_A
a'$ if $\Gamma \vdashS a = a' : A$. 
\end{itemize}
is the free cwf on one base type.
\end{theorem}

By \emph{free cwf on one base type} we mean that it is initial in the
category $\Cwf^o$ having as objects small cwfs with a chosen type $o \in
\Ty(1)$, and as morphisms the strict cwf-morphisms which preserve the
chosen base type. We refer to Castellan, Clairambault, and Dybjer
\cite{castellan:tlca2015,castellan:lmcs} for  the proof of freeness
of $\T$. 

In the papers mentioned above, we show that the free
cwf is also \emph{bifree} in the fully dependent version $\Cwf_{\dem}^2$of the
$2$-category $\Scwf_{\dem}^2$ of Section
\ref{subsubsec:biequivalence_st} (with a base type). Before we
consider additional structure, let us define the morphisms of this $2$-category,
which will play an important role later on.

\begin{definition}\label{def:pseudomor}
  A \emph{pseudo-cwf morphism} from a cwf $\C$ to a cwf $\D$ is a
  pair $(F, \sigma )$ where $F : \C \rightarrow \D$ is a functor and
  for each $ \Gamma \in \C_0$, $ \sigma _ \Gamma $ is a
  \textbf{Fam}-morphism from $T \Gamma $ to $T'F \Gamma $ preserving
  the structure up to isomorphism. In particular there are
isomorphisms (again writing $F$ for all components):
\[
\begin{array}{rrcl}
  \theta _{A,  \gamma } :& FA[F \gamma] &\cong_{F\Gamma}&
F(A[\gamma])\hspace{30pt}\text{(for $\gamma :
\Gamma \to \Delta$)}\\
   !_{F} :& 1 &\cong& F1      \\
  \rho _{ \Gamma , A} :& F( \Gamma \cext A)  &\cong&  F \Gamma \cext FA                                                
\end{array}
\]
where $\cong_{F\Gamma}$ means that it is an isomorphism in the
category $\Ty^\D(F\Gamma)$ of types over $F\Gamma$ in $\D$. These
data must satisfy some coherence diagrams (see
\cite{ClairambaultD14}, Definition 3.1 for details). 
\end{definition}

We can now prove the following.

\begin{theorem}
The cwf $\T$ is \emph{bifree over one base type}, \emph{i.e.} it is
bi-initial in the $2$-category $\Cwf_o^2$ having as objects small cwfs with one base
type, as $1$-cells pseudo cwf-morphisms preserving the base type up
to iso, and as $2$-cells natural transformations between the base
functors.
\end{theorem}
\begin{proof}
We refer the reader to \cite{castellan:lmcs} for details and proofs. We remark that
$\T$ is democratic and also bi-initial in $\Cwf_{o,\dem}^2$.
\end{proof}

Previously we defined initial ucwfs and scwfs (with extra structure) in two ways: with and without explicit substitutions. The above construction of the free cwf gives rise to a calculus with explicit substitutions. There is an analogous construction of a free cwf where substitution is instead defined as a meta-operation, which is close to the standard formulation of the general rules for dependent type theory. We do not have space to spell out this construction with implicit substitutions here, but refer to Streicher \cite{streicher:semtt}.

\subsection{Extensional identity types, $\Sigma$-types, and finite limits}
\label{subsec:sigma}

\subsubsection{Extensional identity types and $\Sigma$-types}

We now add extensional identity types and $\Sigma$-types to cwfs.

\begin{definition}
An $\Iext$-structure on cwf $\C$ consists of
for each $\Gamma \in \C_0$, $A\in \Ty(\Gamma)$ and $a, a' \in
\Tm(\Gamma, A)$, a type $\Id_A(a, a')$; and a term
\[
\refl_{A, a} \in \Tm(\Gamma,\Id_A(a, a))
\]
such that if $c \in \Tm(\Gamma,\Id_{A}(a,a'))$ then $a=a'$ and $c=
\refl_{A, a}$, and such that
\[
\Id_A(a,a')[\gamma] = \Id_{A[\gamma]}(a[\gamma], a'[\gamma])
\]
for any $\gamma \in \C(\Delta, \Gamma)$.
\end{definition}

This captures \emph{extensional}, rather than \emph{intensional},
identity types. The difference is that the two equations $a=a'$ and $c=\refl_{A, a}$
whenever $\Tm(\Gamma,\Id_{A}(a,a'))$ is inhabited are only valid for extensional identity types. It
follows that the reflexivity term is preserved by substitution as
well: $\refl_{A,a}[\gamma]$ is forced to coincide with
$\refl_{A[\gamma], a[\gamma]}$ since the two both inhabit
$\Id_{A[\gamma]}(a[\gamma], a[\gamma])$.

As in the previous sections, strict cwf-functors between cwfs equipped with an
extensional identity type structure are said to \emph{preserve it strictly} if 
the action on morphisms maps the components of the source to the components of
the target, on the nose. In the remainder of this paper, a more important role
will be played by morphisms preserving this structure up to isomorphism.
If $F : \C \to \D$ is a pseudo cwf-morphism where $\C$ and $\D$ are
equipped with an $\Iext$-structure, we say that it
\emph{preserves} it provided there is, for any $A\in \Ty^\C(\Gamma)$,
$a, a' \in \Tm^\C(\Gamma, A)$, an isomorphism
\[
F(\Id_{A}(a, a')) \cong \Id_{FA}(Fa, Fa')
\]
in $\Ty^\D(\Gamma)$. 


%


\begin{definition}
The definition of an $\None$-structure for a cwf is the same as
for an scwf, except that $\None$ now is also required to be stable
under substitution, \emph{i.e.}, a natural transformation of
type presheaves
$$
1 \natto \Ty(-)
$$
and, as before, a natural isomorphism between preheaves
$$
1 \cong \Tm_\C(-,\None)
$$
where again $1$ is the constant singleton presheaf.
\end{definition}

If a cwf is democratic it has an $\None$-structure, since
$\None$ may be defined as $\overline{1}\in \Ty(1)$.

\begin{definition}
A \emph{$\Sigma$-structure} on a cwf $\C$ consists of, for each
$\Gamma \in \C_0$, $A \in \Ty(\Gamma)$, $B \in \Ty(\Gamma \cext A)$,
a type $\Sigma(A, B) \in \Ty(\Gamma)$, and term formers
\begin{eqnarray*}
\fst_{\Gamma, A, B}(-) &:& \Tm(\Gamma, \Sigma(A, B)) \to \Tm(\Gamma,
A)\\
\snd_{\Gamma, A, B}(-) &:& \Pi_{c \in \Tm(\Gamma, \Sigma(A, B))} \Tm(\Gamma,
B[\tuple{\id, \fst(c)}])\\
\tuple{-, -} &:& (\Sigma_{a \in \Tm(\Gamma, A)} 
\Tm(\Gamma, B[\tuple{\id, a}])) \to \Tm(\Gamma, \Sigma(A, B))
\end{eqnarray*}
%
%
%
subject to the same equations as in Definition
\ref{def:times_struct}, plus the additional 
\[
\Sigma(A,B)[\gamma] = \Sigma(A[\gamma],B[\tuple{\gamma\circ
\p,\q}])\,.
\]
\end{definition}

A $\Sigma$-type structure thus gives rise to a natural transformation $\Sigma$ of
type presheaves
$$
\sum_{A \in \Ty_\C(\Gamma)} \Ty_\C(\Gamma\cext A) \natto \Ty_\C(\Gamma)
$$
and isomorphisms
$$
\sum_{a \in \Tm_\C(\Gamma,A)} \Tm_\C(\Gamma,B[ \langle \id, a \rangle ]) \cong \Tm_\C(\Gamma,\Sigma_\Gamma(A,B))
$$
which are stable under substitution (see Definition \ref{def:times_struct}).
Note the difference between this and the characterization of a $\times$-structure as natural isomorphisms between presheaves. Since $A \in \Ty(\Gamma)$ and $B\in \Ty(\Gamma\cext A)$ are dependent types the two sides of the isomorphism are no longer presheaves.

It follows by induction on the length of a context that a \emph{contextual} cwf with an $\None$-structure and a
$\Sigma$-structure is democratic.

As usual, a strict cwf-morphism between cwfs with $\Sigma$-structure
\emph{preserves} it if it maps the structure in the source cwf to the structure
in the target cwf on the nose. For \emph{pseudo}-morphisms, it turns out that
there is nothing to add. In any cwf $\C$ with a $\Sigma$-structure, for any
$\Gamma \in \C_0$, $A\in \Ty(\Gamma)$ and $B \in \Ty(\Gamma \cext A)$, we have
the isomorphism $\Gamma \cext A \cext B \cong \Gamma \cext \Sigma(A, B)$. Since
pseudo cwf-morphisms are already known to preserve context extension, it
follows that they automatically preserve $\Sigma$-structures, in the following
sense. 

\begin{proposition}\label{prop:preservsigma}
A pseudo cwf-morphism $F$ from $\C$ to $\D$, where both cwfs have a
$\Sigma$-structure, also preserves it in the sense that there is an
isomorphism:
\[
s_{A, B} : F(\Sigma(A, B)) \cong \Sigma(FA, FB[\rho_{\Gamma, A}^{-1}])
\]
such that projections and pairs are preserved, modulo some transports (notably
following $s_{A, B}$, see Proposition 3.5 in \cite{ClairambaultD14} for
details). 
\end{proposition}

Let us write $\Cwf_\dem^{2,\Iext,\Sigma}$ for the $2$-category having as
objects small democratic cwfs with an $\Iext$-structure and a
$\Sigma$-structure; as morphisms the pseudo cwf-morphisms preserving these structures
up to isomorphism, and as $2$-cells natural transformations between
the base functors. 

\subsubsection{The biequivalence with finitely complete categories}
\label{subsubsec:bieq_fl}

In Theorem \ref{th:biequivalence_cc}, we proved a biequivalence
between democratic scwfs and cartesian categories. We now sketch the
dependently typed version: a biequivalence between democratic cwfs with
$\Iext$- and $\Sigma$-structures; and \emph{finitely complete}
(also called \emph{left exact}) categories. 

\begin{definition}
A category $\C$ is \emph{finitely complete} if it has all finite limits. A
functor $F : \C \to \D$ between finitely complete categories is \emph{left
exact} if it preserves finite limits: the image of a limiting
cone is a limiting cone.

We write $\FL$ for the $2$-category with small finitely complete categories
as objects, left exact functors as $1$-cells, and natural transformations as
$2$-cells. 
\end{definition}

In the light of Section \ref{sec:scwfs}, it is natural to insist that we
consider finitely complete categories to be categories with \emph{property},
rather than with additional structure. How could we make a corresponding
notion of finitely complete categories with structure? One could ask for a
cartesian category with structure additionally equipped with a choice of
equalizers; one could ask for a choice of pullbacks (and a terminal object); or
one could directly ask for a choice of a limit of any finite diagram. One could
then consider a category of these, where structure is preserved on the nose.
However, the strict equivalence of Theorem \ref{th:eq_strict_cc} does
\emph{not} to extend to this situation: whereas in the simply-typed case we may
prove an equivalence with structure or a biequivalence with property, it seems
that the only possibility to relate the two in the present case is the
biequivalence. We will comment again on that later.

Let us now give some information about the main ingredients of the biequivalence.
The first observation is that if $\C$ is a cwf with $\Iext$ and $\Sigma$-structures, then for each $\Gamma \in \C_0$ there is an equivalence 
between the category of types over a context $\Ty(\Gamma)$, and the
\emph{slice category} $\C/\Gamma$. Indeed, each type $A \in \Ty(\Gamma)$ yields
a \emph{display map} $\p_{\Gamma, A} : \Gamma \cext A \to \Gamma$ regarded as
an object in $\C/\Gamma$. In the other direction, any $\gamma : \Delta \to
\Gamma$ is isomorphic (in $\C/\Gamma$) to
\[
\p_{\Gamma,\Inv(\gamma)} : \Gamma \cext \Inv(\gamma) \to \Gamma\,,
\]
a projection corresponding to a type $\Inv(\gamma) \in \Ty(\Gamma)$, the
``inverse image'', defined as (a cwf formalization of) $x : \Gamma \vdash
\Sigma_{y:\Delta} (\gamma(x) = y)\mathrm{~type}$. Via this equivalence of
categories it follows that for each $\Gamma \in \C_0$ the slice category
$\C/\Gamma$ has products, that is, $\C$ has pullbacks. Since it has a
terminal object, it has all finite limits. Likewise, if $F : \C \to \D$ is a
$1$-cell in $\Cwf_\dem^{2,\Iext,\Sigma}$ then it preserves
pullbacks in $\C$ -- in fact, we have an equivalence:

\begin{lemma}
Let $\C$ and $\D$ be democratic cwfs with $\Iext$- and $\Sigma$-structures and
$F : \C \to \D$ be a pseudo cwf-morphism preserving democracy. Then, $F$
preserves the $\Iext$-structure if and only if $F$ preserves pullbacks.
\end{lemma}
\begin{proof}
The harder direction is \emph{only if}, which boils down to 
the preservation of the inverse image. This can be proved from intricate
calculations on cwf combinators. Details appear in \cite{ClairambaultD14}, Lemma 4.3 and Proposition 4.4. 
\end{proof}

We can then show that there is a \emph{forgetful $2$-functor}:
\[
\CF : \Cwf_\dem^{2,\Iext,\Sigma} \to \FL\,.
\]

The other direction is much more complicated. The equivalence of categories
$\Ty(\Gamma) \simeq \C/\Gamma$ together with Seely's approach to interpreting
type theory in locally cartesian closed categories \cite{SeelyRAG:locccc}
suggest, from a finitely complete category $\C$, to redefine the \emph{types}
over $\Gamma$ as the objects of the slice category $\C/\Gamma$. Now the question is how to define \emph{substitution}, \emph{i.e.}
\[
-[\gamma] : (\C/\Gamma)_0 \to (\C/\Delta)_0
\]
for $\gamma : \Delta \to \Gamma$? In categorical logic,
substitution is usually defined \emph{by pullback}. However, the problem is that
for an arbitrary choice of pullbacks, there is no reason why this assignment should
be functorial. Consider the following two pullback squares:
\[
\xymatrix@R=10pt{
\cdot	\ar[r]
	\ar[d]
	\pb{315}&
\cdot	\ar[r]
	\ar[d]
	\pb{315}&
\Gamma\cext A
	\ar[d]^{\p_{\Gamma, A}}\\
\Omega	\ar[r]^{\delta}&
\Delta	\ar[r]^{\gamma}&
\Gamma
}
\qquad
\qquad
\xymatrix@R=10pt{
\cdot	\ar[r]
	\ar[d]
	\pb{315}&
\Gamma\cext A
	\ar[d]^{\p_{\Gamma, A}}\\
\Omega	\ar[r]^{\gamma \circ \delta}&
\Gamma
}
\]
There is no reason why the left hand side diagram, which is a composition of chosen pullbacks, and the right
hand side diagram, which is a chosen pullback, should coincide -- although they
are always isomorphic. In other words the codomain fibration is not
\emph{split}, whereas the fibration implicit in a cwf is always
split. This is a fundamental issue: Seely's proposed interpretation
\cite{SeelyRAG:locccc} sends types that are provably equal in the
syntax to morphisms in $\C$ that are only known to be isomorphic. This
\emph{coherence problem} may be solved in two ways: Curien
\cite{curien:fi} proposes to change the syntax by weakening equality
to isomorphism \emph{in the syntax}, enriching it with explicit
coercions between isomorphic types, and showing the extended syntax
equivalent to the original via a difficult coherence theorem. We refer to Curien, Garner, and Hofmann \cite{CurienGH14} for details.

In \cite{hofmann:csl}, Hofmann proposes instead to solve the problem by
exploiting a construction of B\'enabou \cite{BenabouJ:fibcfn}, which associates to each fibration an
equivalent split fibration. This construction can be extended to dependent types: given a category
$\C$ with finite limits, we build a cwf whose types are no longer just
objects of $\C/\Gamma$, but objects of $\C/\Gamma$ with a pre-chosen
substitution pullback, for every possible substitution -- such that this choice
is split. For details, the reader is referred to \cite{ClairambaultD14}, Section 5. As we
show there, Hofmann's construction yields a \emph{pseudofunctor}:
\[
\LF : \FL \to \Cwf_\dem^{2,\Iext,\Sigma}\,.
\]

This is \emph{not} a functor: when sending $F : \C \to \D$ to
$\Cwf_\dem^{2,\Iext,\Sigma}$, one must extend $F$ to types, \emph{i.e.}
display maps $\p$ together with a fixed choice of substitution pullbacks.
But the choice of substitution pullbacks for $\p$ in $\C$ does not suffice to completely
determine a choice of substitution pullbacks for $F\p$ in $\D$, hence those must
be \emph{chosen}; causing $\LF$ to fail functoriality on the nose.

For this reason, it seems unlikely that switching to finite limit categories
with structure would yield an equivalence of categories with strict maps,
unless one considers categories with finite limits with a split choice of
pullbacks\footnote{One reviewer pointed out that one should not attribute the
weakness of the equivalence solely to the pseudo-functoriality of $\LF$.
Indeed, the reviewer suggests that a version of $\LF$ based on Lumsdaine-Warren's
``left adjoint splitting'' \cite{DBLP:journals/tocl/LumsdaineW15} rather than
B\'enabou-Hofmann's construction might be functorial, yet still only yield a
\emph{bi}equivalence.}.

\begin{theorem}\label{th:biequivalence_fl}
There is a biequivalence of $2$-categories $\FL^2 \simeq
\Cwf_\dem^{2,\Iext,\Sigma}$.
\end{theorem}
\begin{proof}
Once the mediating pseudofunctors are constructed, the proof is fairly close
to that of Theorem \ref{th:biequivalence_cc}. The reader is referred
to \cite{ClairambaultD14}, Section 6 for details.
\end{proof}

\subsection{$\Pi$-types and locally cartesian closed categories}

We now add $\Pi$-types and extend the results of Section \ref{subsec:sigma}. We shall
sketch that Theorem \ref{th:biequivalence_fl} yields a biequivalence with locally
cartesian closed categories. 

\subsubsection{$\Pi$-types}


\begin{definition}
A $\Pi$-type structure on a cwf $\C$ consists of, for each $\Gamma \in \C_0$,
$A \in \Ty(\Gamma)$, $B \in \Ty(\Gamma \cext A)$, a type $\Pi(A, B) \in
\Ty(\Gamma)$, and term formers
\begin{eqnarray*}
\lambda_{\Gamma, A, B} &:& \Tm(\Gamma\cext A, B) \to \Tm(\Gamma, \Pi(A, B))\\
\ap_{\Gamma, A, B} &:& \Tm(\Gamma, \Pi(A, B)) \to \Pi_{a \in \Tm(\Gamma, A)}
\Tm(\Gamma, B[\tuple{\id, a}]
\end{eqnarray*}
subject to the equations of Definition \ref{def:tto_structure}, plus the
additional 
\[
\Pi(A, B)[\gamma] = \Pi(A[\gamma], B[\tuple{\gamma \circ \p, \q}])\,.
\]
\end{definition}

We note that a $\Pi$-type structure gives rise to a natural transformation $\Pi$ of type presheaves:
$$
\Gamma \mapsto \sum_{A \in \Ty_\C(\Gamma)} \Ty_\C(\Gamma\cext A) \natto \Ty_\C(\Gamma)
$$
and isomorphisms
$$
\Tm_\C(\Gamma\cext A,B) \cong \Tm_\C(\Gamma,\Pi_\Gamma(A,B))
$$
which are stable under substitution (see Definition \ref{def:tto_structure}).
Note the difference between this and the characterization of an $\arrow$-structure as
natural isomorphisms between presheaves. Since $A \in \Ty(\Gamma)$ and $B\in
\Ty(\Gamma\cext A)$ are dependent types $\Tm_\C(\Gamma\cext A,B)$ and $\Tm_\C(\Gamma,\Pi_\Gamma(A,B))$ are no longer families of presheaves.

Strict cwf-morphisms between cwfs with $\Pi$-structure \emph{preserve} it if
they map all components of the $\Pi$-structure in the source cwf to the same
component in the target cwf, on the nose. For \emph{pseudo-morphisms}, we
define:

\begin{definition}
A pseudo cwf-morphism $F$ from $\C$ to $\D$, where both cwfs have a
$\Pi$-structure, also \emph{preserves} it provided for all $\Gamma \in \C_0$,
$A\in \Ty(\Gamma)$ and $B \in \Ty(\Gamma \cext A)$ there is an isomorphism in
$\Ty^\D(F\Gamma)$
\[
i_{A, B} : F(\Pi^C(A,B)) \cong \Pi^{\D}(FA, FB[\rho_{\Gamma, A}^{-1}])
\]
such that application is preserved, modulo some transports (notably by $i_{A,
B}$, see Definition 9 in \cite{ClairambaultD14} for details).
\end{definition}

It is sufficient to require preservation of application, preservation of
abstraction then follows. This follows the situation in cartesian closed
categories with structure where evaluation is part of the
structure, whereas abstraction is defined uniquely by the universal property --
and although functors are only required to preserve evaluation, it follows that
they preserve abstraction too.

Let us write $\Cwf_\dem^{2,\Iext,\Sigma,\Pi}$ for the $2$-category where the
objects are small democratic cwfs with $\Iext$-structure,
$\Sigma$-structure and $\Pi$-structure; the morphisms are pseudo cwf-morphisms
preserving this structure (up to isomorphism), and the $2$-cells are natural
transformations between the base functors. 

\subsubsection{The biequivalence with locally cartesian closed categories}

Let us first recall locally cartesian closed categories.

\begin{definition}
A category $\C$ is \emph{locally cartesian closed (lccc)} if it has a
terminal object and if for all $\Gamma \in \C_0$, the slice category
$\C/\Gamma$ is cartesian closed.  
\end{definition}

Again, this definition is in terms of property rather than structure. This is
one of the two usual definitions of locally cartesian closed categories.
Equivalently, one could ask $\C$ to have finite limits, and require that for
all $\gamma : \Delta \to \Gamma$, the \emph{pullback functor} $\gamma^* :
\C/\Gamma \to \C/\Delta$, which associates to any $f : \cdot \to \Delta$ the
left hand side morphism of the pullback diagram
\[
\xymatrix{
\cdot	\ar[r]
	\ar[d]_{\gamma^*(f)}
	\pb{315}&
\cdot 	\ar[d]^f\\
\Delta	\ar[r]^\gamma&
\Gamma
}
\]
obtained via finite limits, has a right adjoint $\Pi_\delta :
\C/\Delta \to \C/\Gamma$. It is this right adjoint that was proposed by Seely
for the interpretation of $\Pi$-types.

As usual, this right adjoint to the pullback functor may be equivalently
described as the data of cofree objects. Instantiated here, a right
adjoint to $\gamma^* : \C/\Gamma \to \C/\Delta$ consists of, for all
$\delta : \cdot \to \Delta$ an object of $\C/\Delta$, an
object $\Pi_\gamma(\delta) : \cdot \to \Gamma$ in $\C/\Gamma$
together with a co-unit $\epsilon_\delta :
\gamma^*(\Pi_\gamma(\delta)) \to \delta$ (a morphism in $\C/\Delta$)
satisfying the universal property of co-free objects. These data,
along with the universal property, amount to a \emph{dependent
product diagram}\footnote{It was pointed out by a reviewer that those were called
``distributivity pullbacks'' by Weber \cite{weber}}: 
\[
\xymatrix{
&\cdot  \ar@/_/[dl]^{\epsilon_\delta}
        \ar[d]
        \ar[r]
        \pb{315}&
\cdot   \ar[d]^{\Pi_\gamma(\delta)}\\
\cdot   \ar[r]_\delta&
\Delta  \ar[r]_\gamma&
\Gamma
}
\]
which is universal among any such diagram over $\delta$ and $\gamma$, as
described below.
\[
\xymatrix{
&\cdot  \ar@/_/[ddl]
        \ar@/_/[dd]
        \ar@{.>}[d]
        \ar[r]
        \pb{315}&
\cdot   \ar@/_/[dd]
        \ar@{.>}[d]\\
&\cdot  \ar@/_/[dl]^{\epsilon_\delta}
        \ar[d]
        \ar[r]
        \pb{315}&
\cdot   \ar[d]^{\Pi_\gamma(\delta)}\\
\cdot   \ar[r]_\delta&
\Delta  \ar[r]_\gamma&
\Gamma
}
\]

In other words, a locally cartesian closed category may be
defined as a category $\C$ with finite limits such
that, additionally, for every $\cdot \stackrel{\delta}{\to} \Delta
\stackrel{\gamma}{\to} \Gamma$ in $\C$ there is a dependent product diagram
as above. A \emph{locally cartesian closed functor} may then be
defined as a left exact functor (the image of a terminal object is
terminal and the image of a pullback diagram is a pullback diagram)
such that the image of a dependent
product diagram is a dependent product diagram. There is a
$2$-category $\LCC^2$ having small lcccs as objects, locally cartesian
closed functors as $1$-cells, and natural transformations as
$2$-cells. 

Now we can build our biequivalence. First, we define a \emph{forgetful $2$-functor} 
\[
\CF : \Cwf_\dem^{2,\Iext,\Sigma,\Pi} \to \LCC^2\,,
\]
which omits all components and only
keeps the base category, as in Section \ref{subsubsec:bieq_fl}. We must prove that if $\C$ is the base
category of a democratic cwf with 
$\Iext$-structure, $\Sigma$-structure and $\Pi$-structure, then $\C$
is locally cartesian closed. This is straightforward: using
$\Pi$-types it is easy to show that for each $\Gamma$, $\Ty(\Gamma)$
is cartesian closed, but $\Ty(\Gamma)$ is equivalent to $\C/\Gamma$
as shown in Section \ref{subsubsec:bieq_fl}. Alternatively one may
construct dependent products: from $\cdot \stackrel{\delta}{\to}
\Delta \stackrel{\gamma}{\to} \Gamma$ one may construct an isomorphic
(in the obvious sense) sequence of projections via inverse
image:
\[
\Gamma \cext \Inv(\gamma) \cext \Inv(\delta)
\stackrel{\p}{\to} \Gamma
\cext \Inv(\gamma) \stackrel{\p}{\to} \Gamma\,.
\]

For any such sequence of projections $\Gamma \cext A \cext B \to
\Gamma \cext A \to \Gamma$ there is a dependent product diagram,
called the \emph{chosen dependent product diagram}:
\[
\xymatrix{
&\Gamma\cext A\cext \Pi(A, B)[\p_A]
        \ar@/_/[dl]_{\ev_{A, B}}
        \ar[d]^{\p_{\Pi(A, B)[\p_A]}}
        \ar[rr]
        \pb{315}&&
\Gamma\cext \Pi(A, B)
        \ar[d]^{\p_{\Pi(A, B)}}\\
\Gamma\cext A \cext B
        \ar[r]^{\p_B}&
\Gamma\cext A
        \ar[rr]^{\p_A}&&
\Gamma
}
\]
where $\ev_{A, B} = \tuple{\p, \ap(\q, \q[\p])}$. Combined with
inverse image, this shows that
any $\cdot \stackrel{\delta}{\to} \Delta \stackrel{\gamma}{\to}
\Gamma$ has a dependent product diagram. Dependent product diagrams
also permit a nice characterisation of the preservation of
$\Pi$-structures:

\begin{lemma}\label{lem:eq_pres_pi}
Let $F : \C \to \D$ be a pseudo cwf-morphism between cwfs with
$\Pi$-structure. Then $F$ preserves the $\Pi$-structure if and only if
the image of any chosen dependent product diagrams is a dependent product
diagram.
\end{lemma}
\begin{proof}
Proved through intricate calculations -- see \cite{ClairambaultD14}, Proposition 4.8, for details.
\end{proof}

This completes the definition of the forgetful $2$-functor $\CF :
\Cwf_\dem^{2,\Iext,\Sigma,\Pi} \to \LCC^2$. In the other direction,
the construction is the same as for Theorem
\ref{th:biequivalence_fl}; the fact that the pseudofunctor $\LF$
yields pseudo cwf-morphisms preserving $\Pi$-structures follows from
Lemma \ref{lem:eq_pres_pi}. We conclude:

\begin{theorem}\label{th:biequivalence_lcc}
There is a biequivalence of $2$-categories $\LCC^2 \simeq
\Cwf_\dem^{2,\Iext,\Sigma,\Pi}$.
\end{theorem}

\subsubsection{Undecidability in the bifree locally cartesian closed category}

The construction of a free cwf in Theorem \ref{th:free_cwf} can easily be extended
when we add $\Iext, \None, \Sigma$ and $\Pi$-types. This is done by adding the per-rules corresponding to formation, introduction, elimination, and equality rules for $\Iext, \None, \Sigma$ and $\Pi$. We do not have room for explicitly displaying those rules, but they can be found in Castellan, Clairambault, and Dybjer \cite{castellan:lmcs}. We can thus construct a free cwf $\Tlccc$ with $\Iext, \None, \Sigma$ and $\Pi$-type structures. It is both a \emph{free} cwf (with the extra structure) on one
base type and with respect to strict morphisms, and a \emph{bifree}
cwf (with the extra structure) on one base type and with respect to
pseudo morphisms.

\begin{theorem}\label{th:biinitial_cwf}
The cwf $\Tlccc$ is \emph{initial} in the category $\Cwf^{\Iext,\None, \Sigma,
 \Pi, o}$, as well as \emph{bi-initial} in the $2$-category $\Cwf^{2,\Iext,\None,
\Sigma, \Pi, o}$. Moreover, it is democratic and bi-initial in $\Cwf_\dem^{2,\Iext,
\Sigma, \Pi, o}$.
\end{theorem}
\begin{proof}
Details appear in \cite{castellan:lmcs}.
\end{proof}

An object $I$ is
\emph{bi-initial} in a $2$-category iff for any
$A$ there is an arrow $I \to A$ and for any two arrows
$f, g : I \to A$ there exists a unique 2-cell
$\theta : f \Rightarrow g$. It follows that $\theta$ is invertible,
and that bi-initial objects are equivalent.
In the statement above, $o$ denotes a base type. The objects of both 
$\Cwf^{\Iext, \None, \Sigma, \Pi, o}$ and $\Cwf^{2, \Iext, \None, \Sigma, \Pi, o}_\dem$ have a
distinguished type $o \in \Ty(1)$, preserved on the nose for
$\Cwf^{\Iext, \None, \Sigma, \Pi, o}$ and up to isomorphism for
$\Cwf^{2, \Iext, \None, \Sigma, \Pi, o}_\dem$. 

The presence of the base type $o$ does not affect the biequivalence
in Theorem \ref{th:biequivalence_lcc}, which extends to a
biequivalence $\Cwf^{2, \Sigma, \Pi, o}_\dem \simeq \LCC^{2,o}$ where
the latter has a distinguished object $o \in \C_0$ preserved by
functors up to isomorphism. As bi-initial objects are transported to bi-initial
objects via a biequivalence,
it follows from Theorem \ref{th:biinitial_cwf} that
the base category of the cwf $\Tlccc$ is bi-initial in $\LCC^{2, o}$:

\begin{theorem}
The base category of $\Tlccc$ is the bifree lccc on one object.
\end{theorem}

Having constructed a bifree lccc, it is natural to consider its
\emph{word problem}: given two (syntactic) substitutions $\gamma,
\gamma' : \Delta \to \Gamma$ in $\Tlccc$, is it decidable whether they
are equal, or equivalently whether they have the same interpretation
in any locally cartesian closed category? As $\gamma, \gamma'$ are
syntactic constructs in extensional type theory, one expects
undecidability. However, prior undecidability proofs for extensional type
theory rely on structure not available in our case. The folklore argument uses a universe and
Hofmann's undecidability proof \cite{hofmann:thesis} uses a type of natural
numbers. 

In \cite{castellan:lmcs}, we generalize the folklore result to hold
with only one base type, \emph{i.e.} in $\Tlccc$, but without natural numbers or a universe. This relies on an
encoding of combinatory logic in Martin-L\"of type theory
with $\I$-types, $\Pi$-types, and a base type $o$; as a context
$\Gamma_{\mathsf{CL}}$ containing:
\begin{eqnarray*}
k &:& o,\\
s &:& o,\\
\cdot  &:& o \arrow o \arrow o,\\
ax_{k} &:& \Pi x y : o.\ \I(o,\,k\cdot x\cdot y,\, x),\\
ax_{s} &:& \Pi x y z : o.\ \I(o,\, s\cdot x\cdot y\cdot z,\, x\cdot
z\cdot (y\cdot z))
\end{eqnarray*}
where the left-associative binary infix symbol ``$\cdot$'' stands for
application. 
While the above uses the syntax of type theory, it is easy to set up
the same context just using the cwf combinators available in $\Tlccc$,
hence reducing the decision of equality between terms $M, N$ in
combinatory logic to deciding 
\[
\Gamma_{\mathsf{CL}} \vdash M \stackrel{?}{=} N : o
\]
\emph{e.g.} an equality of two terms $M,N \in \Tm(\Gamma_{\mathsf{CL}},
o)$. We conclude:
\begin{theorem}
Equality is undecidable in the bifree locally cartesian closed
category on one base type.
\end{theorem}

\begin{acknowledgement}
Pierre Clairambault was supported by the LABEX MILYON (ANR-10-LABX-0070) of Universit\'e de Lyon,
within the program ``Investissements d'Avenir'' (ANR-11-IDEX-0007) operated by the
French National Research Agency (ANR).

Peter Dybjer was supported by the Centre for Advanced Study at
the Norwegian Academy of Science and Letters in Oslo, where he was a fellow in the
Homotopy Type Theory and Univalent Foundations project while working on this paper.

Finally, we would like to thank the anonymous reviewers for several useful suggestions.
\end{acknowledgement}

\bibliography{}

\begin{thebibliography}{10}

\bibitem{AbadiCCL90}
Mart{\'{\i}}n Abadi, Luca Cardelli, Pierre{-}Louis Curien, and Jean{-}Jacques
  L{\'{e}}vy.
\newblock Explicit substitutions.
\newblock In {\em Conference Record of the Seventeenth Annual {ACM} Symposium
  on Principles of Programming Languages, San Francisco, California, USA,
  January 1990}, pages 31--46, 1990.

\bibitem{aczel:frege}
Peter Aczel.
\newblock {\em Frege Structures and the Notions of Proposition, Truth, and
  Set}, pages 31--59.
\newblock North-Holland, 1980.

\bibitem{AhrensLV18}
Benedikt Ahrens, Peter~LeFanu Lumsdaine, and Vladimir Voevodsky.
\newblock Categorical structures for type theory in univalent foundations.
\newblock {\em Logical Methods in Computer Science}, 14(3), 2018.

\bibitem{Awodey18}
Steve Awodey.
\newblock Natural models of homotopy type theory.
\newblock {\em Mathematical Structures in Computer Science}, 28(2):241--286,
  2018.

\bibitem{barendregt:lambda}
Henk~P. Barendregt.
\newblock {\em The Lambda Calculus}.
\newblock North-Holland, 1984.
\newblock Revised edition.

\bibitem{BenabouJ:fibcfn}
Jean B{\'{e}}nabou.
\newblock Fibred categories and the foundation of naive category theory.
\newblock {\em Journal of Symbolic Logic}, 50:10--37, 1985.

\bibitem{Brilakis18}
Konstantinos Brilakis.
\newblock {On Initial Categories with Families -- Formalization of Unityped and
  Simply Typed CwFs in Agda}.
\newblock Master's thesis, Chalmers University of Technology, 2018.

\bibitem{cartmell:apal}
John Cartmell.
\newblock Generalized algebraic theories and contextual categories.
\newblock {\em Annals of Pure and Applied Logic}, 32:209--243, 1986.

\bibitem{castellan:tlca2015}
Simon Castellan, Pierre Clairambault, and Peter Dybjer.
\newblock Undecidability of equality in the free locally cartesian closed
  category.
\newblock In {\em 13th International Conference on Typed Lambda Calculi and
  Applications, {TLCA} 2015, July 1-3, 2015, Warsaw, Poland}, pages 138--152,
  2015.

\bibitem{castellan:lmcs}
Simon Castellan, Pierre Clairambault, and Peter Dybjer.
\newblock Undecidability of equality in the free locally cartesian closed
  category (extended version).
\newblock {\em Logical Methods in Computer Science}, 13(4), 2017.

\bibitem{ClairambaultD11}
Pierre Clairambault and Peter Dybjer.
\newblock The biequivalence of locally cartesian closed categories and
  {Martin-L{\"{o}}f} type theories.
\newblock In {\em Typed Lambda Calculi and Applications - 10th International
  Conference, {TLCA} 2011, Novi Sad, Serbia, June 1-3, 2011. Proceedings},
  pages 91--106, 2011.

\bibitem{ClairambaultD14}
Pierre Clairambault and Peter Dybjer.
\newblock The biequivalence of locally cartesian closed categories and
  {Martin-L{\"{o}}f} type theories.
\newblock {\em Mathematical Structures in Computer Science}, 24(6), 2014.

\bibitem{clairambault:london}
Pierre Clairambault and Peter Dybjer.
\newblock Game semantics and normalization by evaluation.
\newblock In Andrew~M. Pitts, editor, {\em Foundations of Software Science and
  Computation Structures - 18th International Conference, FoSSaCS 2015, Held as
  Part of the European Joint Conferences on Theory and Practice of Software,
  {ETAPS} 2015, London, UK, April 11-18, 2015. Proceedings}, volume 9034 of
  {\em Lecture Notes in Computer Science}, pages 56--70. Springer, 2015.

\bibitem{Curien86}
Pierre{-}Louis Curien.
\newblock Categorical combinators.
\newblock {\em Information and Control}, 69(1-3):188--254, 1986.

\bibitem{curien:fi}
Pierre-Louis Curien.
\newblock Substitution up to isomorphism.
\newblock {\em Fundamenta Informaticae}, 19(1,2):51--86, 1993.

\bibitem{CurienGH14}
Pierre{-}Louis Curien, Richard Garner, and Martin Hofmann.
\newblock Revisiting the categorical interpretation of dependent type theory.
\newblock {\em Theoretical Computer Science}, 546:99--119, 2014.

\bibitem{dybjer:torino}
Peter Dybjer.
\newblock Internal type theory.
\newblock In {\em TYPES '95, Types for Proofs and Programs}, number 1158 in
  Lecture Notes in Computer Science, pages 120--134. Springer, 1996.

\bibitem{Fiore12}
Marcelo Fiore.
\newblock Discrete generalised polynomial functors.
\newblock Slides for a talk given at the 39th International Colloquium on
  Automata, Languages and Programming (ICALP), 2012.

\bibitem{VoevodskyF17}
Marcelo Fiore and Vladimir Voevodsky.
\newblock Lawvere theories and {C-systems}.
\newblock 2017.
\newblock Accepted for publication in Proceedings of the American Mathematical
  Society.

\bibitem{HofmannM:intttl}
Martin Hofmann.
\newblock Interpretation of type theory in locally cartesian closed categories.
\newblock In {\em Proceedings of CSL}. Springer LNCS, 1994.

\bibitem{hofmann:csl}
Martin Hofmann.
\newblock On the interpretation of type theory in locally cartesian closed
  categories.
\newblock In Leszek Pacholski and Jerzy Tiuryn, editors, {\em CSL}, volume 933
  of {\em Lecture Notes in Computer Science}. Springer, 1994.

\bibitem{hofmann:thesis}
Martin Hofmann.
\newblock {\em Extensional concepts in intensional type theory}.
\newblock PhD thesis, University of Edinburgh, 1995.

\bibitem{hofmann:cambridge}
Martin Hofmann.
\newblock Syntax and semantics of dependent types.
\newblock In Andrew Pitts and Peter Dybjer, editors, {\em Semantics and Logics
  of Computation}. Cambridge University Press, 1996.

\bibitem{HylandP07}
Martin Hyland and John Power.
\newblock The category theoretic understanding of universal algebra: Lawvere
  theories and monads.
\newblock {\em Electr. Notes Theor. Comput. Sci.}, 172:437--458, 2007.

\bibitem{jacobs1999categorical}
Bart Jacobs.
\newblock {\em Categorical logic and type theory}, volume 141 of {\em Studies
  in Logic and the Foundations of Mathematics}.
\newblock Elsevier, 1999.

\bibitem{johnstone2002sketches}
Peter~T Johnstone.
\newblock {\em Sketches of an elephant: A topos theory compendium}, volume~2.
\newblock Oxford University Press, 2002.

\bibitem{LS86}
J.~Lambek and P.~J. Scott.
\newblock {\em Introduction to higher order categorical logic}.
\newblock Number~7 in Cambridge studies in advanced mathematics. Cambridge
  University Press, 1986.

\bibitem{lawvere:hyperdoctrines}
F.~William Lawvere.
\newblock Equality in hyperdoctrines and comprehension schema as an adjoint
  functor.
\newblock In A.~Heller, editor, {\em Applications of Categorical Algebra,
  Proceedings of Symposia in Pure Mathematics}. AMS, 1970.

\bibitem{DBLP:journals/tocl/LumsdaineW15}
Peter~LeFanu Lumsdaine and Michael~A. Warren.
\newblock The local universes model: An overlooked coherence construction for
  dependent type theories.
\newblock {\em {ACM} Trans. Comput. Log.}, 16(3):23:1--23:31, 2015.

\bibitem{martinlof:hannover}
Per Martin-L{\"o}f.
\newblock Constructive mathematics and computer programming.
\newblock In {\em Logic, Methodology and Philosophy of Science, VI, 1979},
  pages 153--175. North-Holland, 1982.

\bibitem{martinlof:padova}
Per Martin-L{\"o}f.
\newblock {\em Intuitionistic Type Theory}.
\newblock Bibliopolis, 1984.

\bibitem{martinlof:gbg92}
Per Martin-L{\"o}f.
\newblock Substitution calculus.
\newblock Notes from a lecture given in {G}{\"o}teborg, November 1992.

\bibitem{Newstead:phd}
Clive Newstead.
\newblock {\em Algebraic Models of Dependent Type Theory}.
\newblock PhD thesis, Department of Mathematical Sciences, Carnegie Mellon
  University, 2018.

\bibitem{Obtulowicz77}
Adam {Obtu\l owicz}.
\newblock Functorial semantics of the type free $\lambda$-$\beta\eta$ calculus.
\newblock In {\em {Foundations of Computation Theory}}, pages 302--307, 1977.

\bibitem{SeelyRAG:locccc}
R.~A.~G. Seely.
\newblock Locally cartesian closed categories and type theory.
\newblock {\em Proceedings of the Cambridge Philosophical Society}, 95:33--48,
  1984.

\bibitem{streicher:semtt}
Thomas Streicher.
\newblock {\em Semantics of Type Theory}.
\newblock {Birkh\"auser}, 1991.

\bibitem{tasistro:lic}
Alvaro Tasistro.
\newblock Formulation of {Martin-L\"of's} theory of types with explicit
  substitutions.
\newblock Technical report, Department of Computer Sciences, Chalmers
  University of Technology and University of {G\"oteborg}, 1993.

\bibitem{trimble:nlab}
Todd Trimble.
\newblock Towards a doctrine of operads.
\newblock Article in nlab
  \url{https://ncatlab.org/toddtrimble/published/Towards+a+doctrine+of+operads#cartesian_operads_are_equivalent_to_lawvere_theories},
  2013 (accessed 8 March 2019).

\bibitem{weber}
Mark Weber.
\newblock Polynomials in categories with pullbacks.
\newblock {\em Theory Appl. Categ}, 30(16):533--598, 2015.

\end{thebibliography}
\end{document}